\documentclass[journal,12pt,onecolumn,draftclsnofoot,]{IEEEtran}

\usepackage{graphicx}
\usepackage{subcaption}
\usepackage[mathscr]{eucal}
\usepackage[cmex10]{amsmath}
\usepackage{epsfig,epsf,psfrag}
\usepackage{amssymb,amsmath,amsfonts,latexsym}
\usepackage{amsmath,graphicx,bm,xcolor,url}
\usepackage{epstopdf}
\usepackage{array}
\usepackage{verbatim}
\usepackage{bm}
\usepackage{algorithmic, cite}
\usepackage{algorithm}
\usepackage{verbatim}
\usepackage{textcomp}
\usepackage{mathrsfs}
\usepackage{amsthm}

\catcode`~=11 \def\UrlSpecials{\do\~{\kern -.15em\lower .7ex\hbox{~}\kern .04em}} \catcode`~=13 

\allowdisplaybreaks[3]

\newcommand{\ba}{\mathbf{a}}
\newcommand{\bA}{\mathbf{A}}

\newcommand{\bB}{\mathbf{B}}
\newcommand{\bc}{\mathbf{c}}
\newcommand{\bC}{\mathbf{C}}

\newcommand{\bI}{\mathbf{I}}

\newcommand{\bk}{\mathbf{k}}
\newcommand{\bK}{\mathbf{K}}

\newcommand{\br}{\mathbf{r}}
\newcommand{\bR}{\mathbf{R}}

\newcommand{\bv}{\mathbf{v}}

\newcommand{\bw}{\mathbf{w}}

\newcommand{\bx}{\mathbf{x}}

\newcommand{\by}{\mathbf{y}}

\DeclareMathAlphabet{\mathbsf}{OT1}{cmss}{bx}{n}
\DeclareMathAlphabet{\mathssf}{OT1}{cmss}{m}{sl}

\DeclareSymbolFont{bsfletters}{OT1}{cmss}{bx}{n}  
\DeclareSymbolFont{ssfletters}{OT1}{cmss}{m}{n}
\DeclareMathSymbol{\bsfGamma}{0}{bsfletters}{'000}
\DeclareMathSymbol{\ssfGamma}{0}{ssfletters}{'000}
\DeclareMathSymbol{\bsfDelta}{0}{bsfletters}{'001}
\DeclareMathSymbol{\ssfDelta}{0}{ssfletters}{'001}
\DeclareMathSymbol{\bsfTheta}{0}{bsfletters}{'002}
\DeclareMathSymbol{\ssfTheta}{0}{ssfletters}{'002}
\DeclareMathSymbol{\bsfLambda}{0}{bsfletters}{'003}
\DeclareMathSymbol{\ssfLambda}{0}{ssfletters}{'003}
\DeclareMathSymbol{\bsfXi}{0}{bsfletters}{'004}
\DeclareMathSymbol{\ssfXi}{0}{ssfletters}{'004}
\DeclareMathSymbol{\bsfPi}{0}{bsfletters}{'005}
\DeclareMathSymbol{\ssfPi}{0}{ssfletters}{'005}
\DeclareMathSymbol{\bsfSigma}{0}{bsfletters}{'006}
\DeclareMathSymbol{\ssfSigma}{0}{ssfletters}{'006}
\DeclareMathSymbol{\bsfUpsilon}{0}{bsfletters}{'007}
\DeclareMathSymbol{\ssfUpsilon}{0}{ssfletters}{'007}
\DeclareMathSymbol{\bsfPhi}{0}{bsfletters}{'010}
\DeclareMathSymbol{\ssfPhi}{0}{ssfletters}{'010}
\DeclareMathSymbol{\bsfPsi}{0}{bsfletters}{'011}
\DeclareMathSymbol{\ssfPsi}{0}{ssfletters}{'011}
\DeclareMathSymbol{\bsfOmega}{0}{bsfletters}{'012}
\DeclareMathSymbol{\ssfOmega}{0}{ssfletters}{'012}

\newcommand{\eqdef}{\stackrel{\Delta}{=}}

\DeclareMathOperator*{\argmax}{arg\,max}

\newtheorem{theorem}{Theorem} 
\newtheorem{lemma}[theorem]{Lemma}

\newtheorem{corollary}[theorem]{Corollary}
\newtheorem{definition}{Definition} 
\newtheorem{example}{Example}

\newtheorem{data model}{Data Model}

\newcommand{\qednew}{\nobreak \ifvmode \relax \else
      \ifdim\lastskip<1.5em \hskip-\lastskip
      \hskip1.5em plus0em minus0.5em \fi \nobreak
      \vrule height0.75em width0.5em depth0.25em\fi}

\DeclareMathOperator{\lcm}{lcm}
\DeclareMathOperator{\opt}{opt}

\begin{document}
\title{Support Recovery of Periodic Mixtures 
with Nested Periodic Dictionaries}
\author{Pouria~Saidi and~George~K. Atia
\thanks{The work was supported in part by NSF CAREER Award CCF-1552497 and NSF Award CCF-2106339.
\\
Pouria Saidi is with the School of Electrical, Computer and Energy Engineering, Arizona State University, Tempe, AZ, 85281 USA.
George K. Atia is with the Department of Electrical and Computer Engineering and the Department of Computer Science, University of Central Florida, Orlando, FL, 32816 USA. (e-mail: psaidi@asu.edu; george.atia@ucf.edu). This work is done while Pouria Saidi was at University of Central Florida.}}

\maketitle
\begin{abstract}
Periodic signals composed of periodic mixtures admit sparse representations in nested periodic dictionaries (NPDs). Therefore, their underlying hidden periods can be estimated by recovering the exact support of said representations. In this paper, support recovery guarantees of such signals are derived both in noise-free and noisy settings.   
While exact recovery conditions have been studied in the theory of compressive sensing, existing conditions fall short of yielding meaningful achievability regions in the context of periodic signals with sparse representations in NPDs, in part since existing bounds do not capture structures intrinsic to these dictionaries. 
We leverage known properties of NPDs to derive several conditions for exact sparse recovery of periodic mixtures in the noise-free setting. These conditions rest on newly introduced notions of nested periodic coherence and restricted coherence, which can be efficiently computed. In the presence of noise, we obtain improved conditions for recovering the exact support set of the sparse representation of the periodic mixture via orthogonal matching pursuit based on the introduced notions of coherence. 
The theoretical findings are corroborated using numerical experiments for different families of NPDs. Our results show significant improvement over generic recovery bounds as the conditions hold over a larger range of sparsity levels.  
\end{abstract}
\begin{IEEEkeywords}
Nested periodic dictionary, Nested periodic coherence, Ramanujan subspaces, Support recovery conditions
\end{IEEEkeywords}
\section{Introduction} \label{sec:intro}
\IEEEPARstart{E}stimating the underlying periods of signals that exhibit periodic patterns is a fundamental task 
in many real-world applications across diverse disciplines, such as astronomy \cite{hewish201374}, healthcare \cite{li2007robust},\cite{da2017matched} and neuro-rehabilitation \cite{lin2006frequency},\cite{saidi2017detection}.
A discrete periodic signal can be expressed as
\begin{equation}
    \label{eq:discrete_period}
    \begin{aligned}
    y\left(n\right) = y\left(n + p\right) \quad \forall n\in \mathbb{Z}
    \end{aligned},
\end{equation}
where $p$ is the smallest integer satisfying \eqref{eq:discrete_period}.
The periodic signal $y\left(n\right)$ could be a periodic mixture that consists of multiple hidden periods. In \cite{tenneti2018minimum}, the authors defined two criteria to characterize the hidden periods. First, the signal can be expressed as a combination of periodic signals $y_i, i = 1\ldots, m$,
\begin{equation}
\label{eq:mixture_definition}
    y\left(n\right) = y_1\left(n\right) + y_2\left(n\right)+\ldots y_m\left(n\right)
\end{equation}
with distinct periods in $\mathbb{P}_h = \{p_1,p_2,\ldots p_m\}$, where these periodic signals with period $p_i$ cannot be further decomposed into sums of periodic signals with periods smaller than the $p_i$'s. Second, the periods $p_1,p_2,\ldots,p_m$ are not divisors of one another, denoted $p_i\not\vert p_j$ for $i,j \in \{1,2,\ldots,m\}$.
It has been shown that the period of $y\left(n\right)$ with hidden periods $\mathbb{P}_h$ is $p = \lcm\left(\mathbb{P}_h\right)$, i.e., the least common multiplier (lcm) of the periods in the set $\mathbb{P}_h$ \cite[Corollary 2]{tenneti2018minimum}. 

The problem of period estimation has been widely studied and many techniques have been proposed to estimate the underlying periods of periodic signals \cite{wise1976maximum},\cite{christensen2007joint},\cite{sethares1999periodicity}. Conventional methods, such as successive differencing, fail to identify the hidden periods in $y\left(n\right)$ when it is composed of a mixture of periodic signals. Also, spectral-based methods such as DFT could work well if the datalength $L$ is a multiple of the period but yield inaccurate estimates of the fundamental frequencies for arbitrary $L$ \cite{tenneti2015nested}. Estimating the hidden periods in the presence of noise is an even more challenging problem. For a comprehensive discussion of the limitations of such techniques, we refer the reader to \cite{tenneti2015nested}.

To overcome these limitations, the authors in \cite{tenneti2015nested} introduced a family of matrices called nested periodic dictionaries (NPDs) to model periodic mixtures, and have shown that periodic signals often admit sparse representations in NPDs. In \cite{tenneti2015nested,tenneti2016unified}, it was shown that one can estimate the hidden periods of a periodic mixture by estimating its sparse representation in an NPD, then the recovered support set of the sparse vector (i.e., the locations of its nonzero entries) can be used to reveal the hidden periods.

NPDs have already found numerous applications in various domains. In \cite{saidi2019detection}, the authors developed a framework for the detection of the so-called steady-state visually evoked potentials (SSVEPs) using the Ramanujan periodicity transform (RPT) dictionary, an instance of NPDs. SSVEPs are periodic electric potentials induced on the brain cortex in response to external visual stimuli, and can be used to develop brain computer interfaces. In  \cite{tenneti2016detecting}, NPDs were used to detect tandem repeats in DNA, and  in \cite{mathur2021ramanujan}, a method based on Ramanujan subspaces is proposed for classification of epileptic EEG signals. 

There is evidence that one can use sparse recovery frameworks to reconstruct the support set of the sparse representations of periodic mixtures in NPDs given that the datalength is sufficiently large, and then estimate the hidden periods based on the recovered support set. This has been shown through numerical simulations in \cite{vaidyanathan2014farey},\cite{tenneti2015nested}. Nevertheless, there are no theoretical guarantees for recovering the exact support set of a sparse representation of a periodic mixture. If the support set of the sparse vector cannot be precisely recovered, it cannot be utilized for identifying the underlying periods. For instance, Fig. \ref{subfig:bp_ex1} shows the recovered coefficients of the sparse representations of a periodic signal with period $20$, $L=200$, and their agreement with those obtained by an oracle where the support of the sparse vector is known a priori, and Fig.~\ref{subfig:bp_ex2} shows that the exact support set of the sparse vector is not recovered successfully when $L = 80$, and as a result identifying the true period of the signal using NPDs fails. This underscores the need to establish conditions that guarantee exact support recovery of periodic mixtures with NPDs. Despite the ability of sparse recovery algorithms to recover said sparse representations in certain regimes, existing sparse recovery guarantees, which are largely based on random dictionaries, fail to explain the observed performance with NPDs (which are deterministic and structured matrices). Thus, these generic guarantees are of limited use in this context,  in part since existing bounds do not capture structures intrinsic to NPDs.

\begin{figure}
    \centering
    \begin{subfigure}[b]{0.46\textwidth}
    \centering
        \includegraphics[width = \textwidth,height = 6cm]{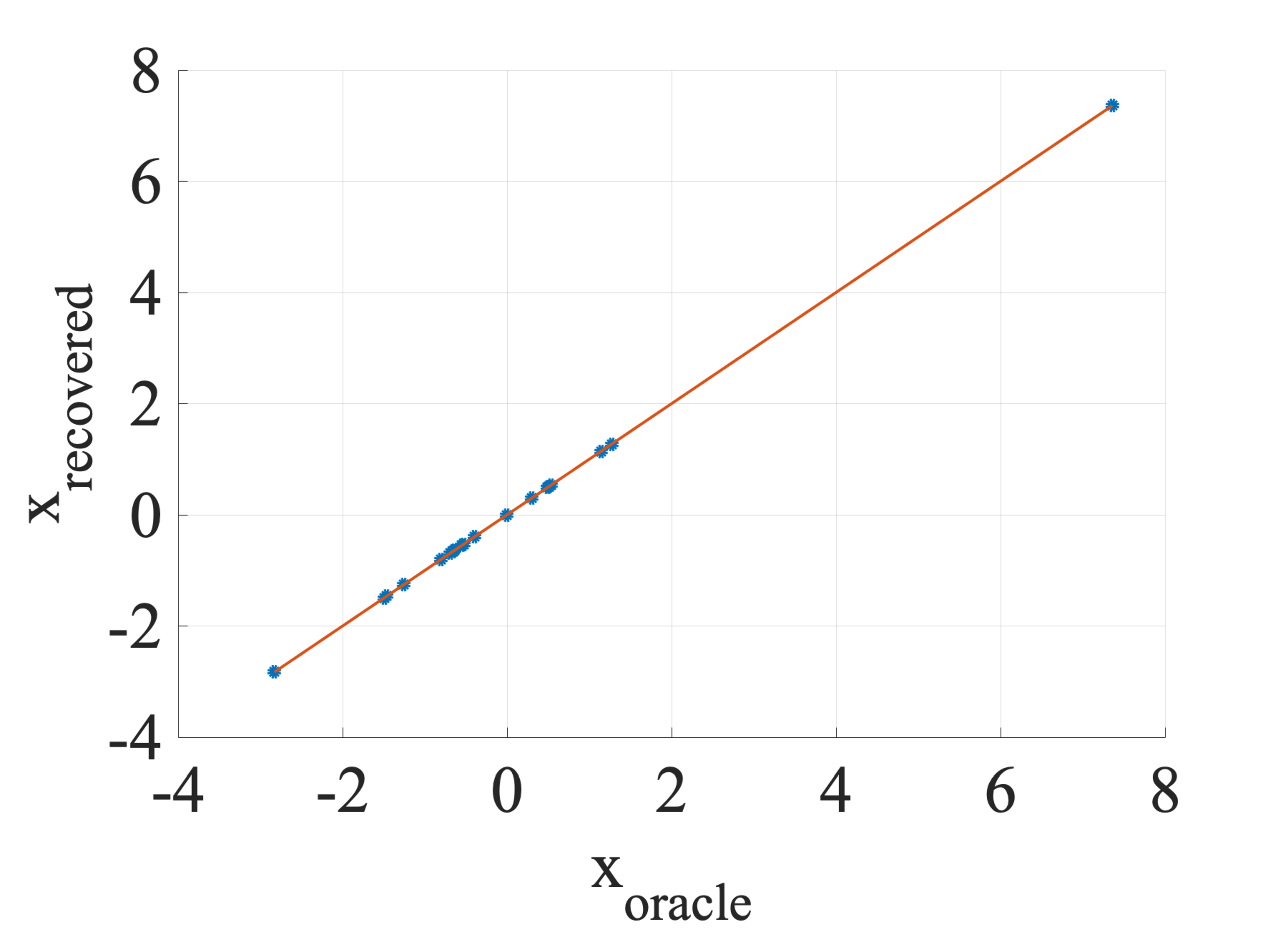}
        \caption{}
        \label{subfig:bp_ex1}
    \end{subfigure}
    \hfill
    \begin{subfigure}[b]{0.46\textwidth}
    \centering
        \includegraphics[width = \textwidth,height = 6cm]{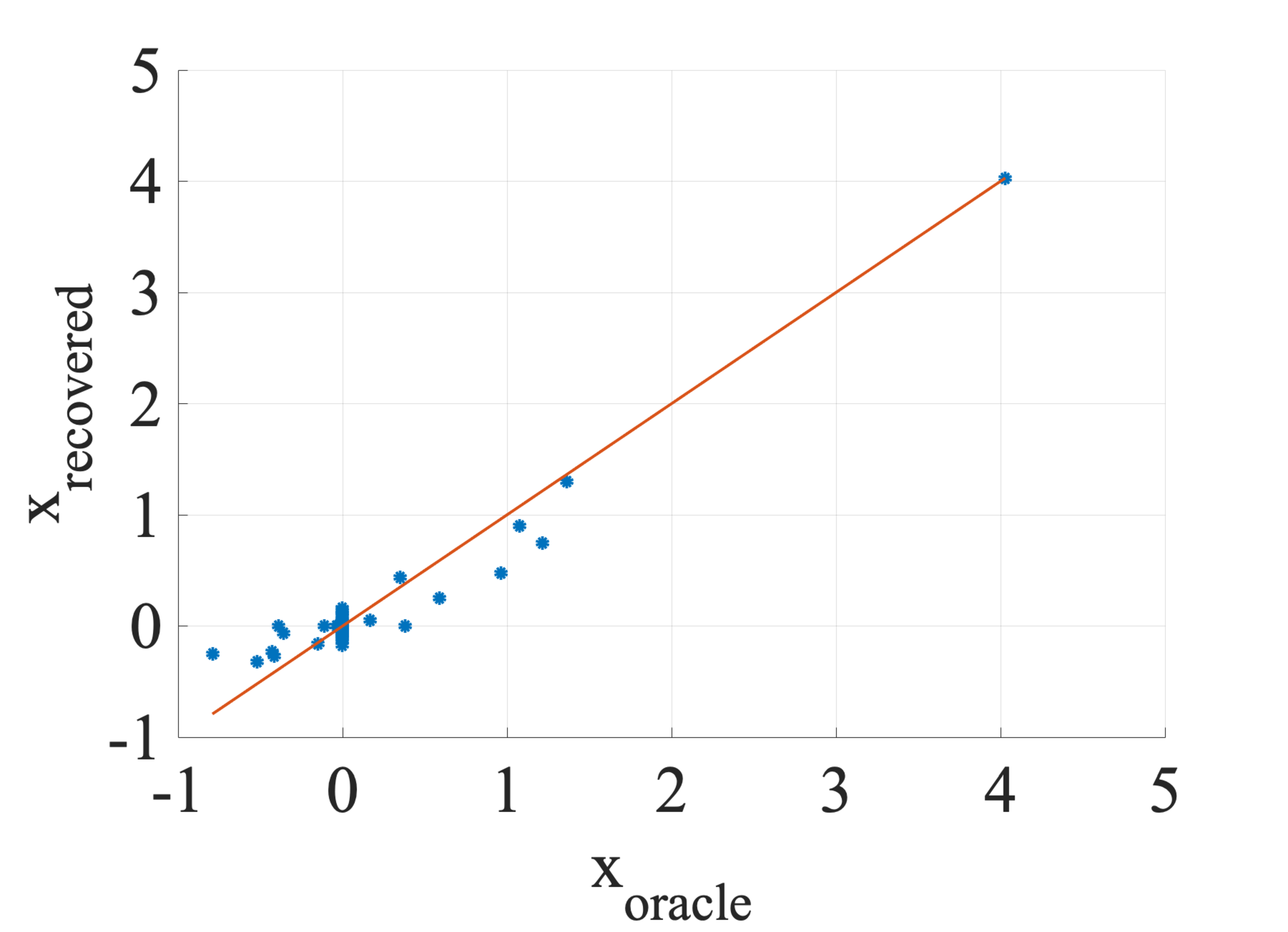}
    \caption{}
    \label{subfig:bp_ex2}
    \end{subfigure}
    \caption{Sparse recovery of periodic signals in NPDs using basis pursuit (a sparse recovery framework) and its agreement with true sparse representation of the signal (a) L = 200 and (b) L = 80.}
    \label{fig:bp_examples}
\end{figure}
In a preliminary version of this work \cite{saidi2021sparse}, we established a new and improved condition for exact sparse recovery with NPDs in the noise-free setting, limited to the special case where $m = 1$, i.e., when the periodic signal consists of one single period. In this work, we extend our study along several fronts. First, we establish exact recovery conditions for periodic mixtures consisting of multiple hidden periods. The conditions rest on newly introduced notions of nested periodic coherence and restricted coherence that leverage structural properties of NPDs.  
Second, we establish new bounds for exact support recovery in presence of noise. 
The following summarizes the main contributions of this paper.
\begin{itemize}
    \item We introduce the new notion of nested periodic coherence, which we leverage to derive a tight upper bound on an exact recovery-type condition for periodic mixtures with sparse representations in an NPD. 
    \item We define the notion of restricted coherence and derive an alternative bound that is computationally more efficient and not restricted to a given sparsity level. 
    \item 
    We derive improved conditions to exactly recover the support set of a periodic mixture in presence of noise for both bounded and Gaussian noise. 
    \item We demonstrate our results for different families of NPDs, including the Farey and RPT dictionaries.
    \end{itemize}

\subsection{Notation and paper organization} We use lowercase letters for scalars, bold lowercase letters for vectors and bold uppercase letters for matrices. For a vector $\bv$ of size $l$ and entries $v_i$, $\|\bv\|_p:=\left(\sum_{i=1}^l |v_i|^p\right)^{1/p}$ denotes its $\ell_p$-norm for $p\geq 1$, $\|\bv\|_0$ is its $\ell_0$-norm which is the number of its nonzero entries, \textcolor{black}{and $\|\bv\|_{\infty} := \max_{i} |v_i|$ denotes its $\ell_{\infty}$-norm.}
For an $L\times N$ matrix $\bA$ with entries $a_{i,j}$, the $\ell_1$-induced norm is $\|\bA\|_{1,1} := \max_{1\leq j\leq N} \sum_{i=1}^L |a_{i,j}|$, $\bA^H$ is its conjugate-transpose, $\bA^{\dagger}=\left(\bA^H\bA\right)^{-1}\bA^H$ is its pseudo-inverse, and $\ba_i$ is its $i$-th column. For a value  $p$, we use $q|p$ to indicate that $q$ is a divisor of $p$. The Euler totient function $\phi\left(p\right)$ of $p$ is the number of positive integers smaller than $p$ that are co-prime to $p$. Set $\mathbb{N} = \{1,\ldots,N\}$ contains all the values between $1$ to $N$. Given a set $S$, the set $S^c$ denotes its complement and $|S|$ its cardinality and $2^S$ denotes its power set that is the collection of all subsets of $S$. We use $\bigoplus_{i = 1}^N \mathcal{G}_i$ to show the direct sum of the subspaces $\mathcal{G}_1,\mathcal{G}_2,\ldots,\mathcal{G}_N$ that is the span of union of subspaces. The matrix $\bK_{S}$ is the matrix $\bK$ restricted to the atoms indexed by set $S$. We use the notation $\bx \sim \mathcal{N}(\mathbf{0}, \sigma^2 \bI_L)$ to show that an $L \times 1$ random vector $\bx$ has a Gaussian distribution, where the mean vector is $\mathbf{0}$ and covariance matrix is $\sigma^2\bI_L$. We denote the minimum eigenvalue of a square matrix $\bA$ with $\lambda_{\min}$.

In Section \ref{sec:background}, we provide a brief overview of relevant sparse recovery techniques and existing guarantees, and present basic background on Ramanujan subspaces and the construction of NPDs. In Section \ref{sec:contribution}, we introduce the new notions of nested periodic coherence and restricted coherence and present the main theorems. We present our numerical results in Section \ref{sec:result}. In Section \ref{sec:discussion}, we present a discussion on the main results and conclude the paper.

\section{Background and preliminaries}
\label{sec:background}
\subsection{Sparse recovery conditions} \label{subsec:sparse_recovery}
Given an underdetermined system of linear equations $\by = \bK\bx$,  one can recover the vector $\bx$ exactly from the measurements $\by$ provided that $\bx$ is sparse, that is, if only a few of its entries are nonzero \cite{donoho2003optimally}. 
A $k$-sparse vector is one with at most $k$ nonzero values. Henceforth, we assume that the columns of the matrix $\bK$ are of unit norm, i.e., $\|\bk_i\|_2 = 1, \forall i\in \mathbb{N}$. The sparse vector $\bx$ can be recovered by minimizing its $\ell_0$-norm subject to the data constraint as, 
\begin{equation}
    \label{eq:l_0_norm}
    \begin{aligned}
        \min_{\bx} \|\bx\|_0 \quad \text{subject to} \quad \by = \bK\bx\:.
    \end{aligned}
\end{equation}
Problem \eqref{eq:l_0_norm} is NP-hard as it requires combinatorial search over all support sets to identify the most sparse solution 
\cite{Baraniuk2007SPM}. 
Therefore, alternative recovery methods were developed, which can be broadly categorized into two main categories: convex relaxation-based and greedy methods. An instance of the former is the popular basis pursuit (BP) in \eqref{eq:l_1_norm}, which uses the $\ell_1$-norm as a surrogate function for the $\ell_0$-norm,  
\begin{equation}
    \label{eq:l_1_norm}
    \begin{aligned}
        \min_{\bx} \|\bx\|_1 \quad \text{subject to} \quad \by = \bK\bx\:.
    \end{aligned}
\end{equation}
A notable instance of the latter is the OMP algorithm which iteratively approximates the signal by sparse orthogonal projections onto the subspace spanned by selected atoms of the dictionary \cite{pati1993orthogonal,davis1994adaptive,tropp2004greed}. Other examples are the ROMP \cite{needell2009uniform} and iterative hard thresholding (IHT) \cite{blumensath2009iterative} algorithms.
Several works established the equivalence of \eqref{eq:l_1_norm} and \eqref{eq:l_0_norm} in that they provably yield the same solution under certain conditions 
\cite{Baraniuk2007SPM},\cite{tropp2004greed},\cite{donoho2006compressed},\cite{candes2005restrictedIP},\cite{candes2008restricted},\cite{cohen2009compressed}.
Among such conditions is the restricted isometry property (RIP) introduced in \cite{candes2005restrictedIP} in which a dictionary $\bK$ satisfies 
\begin{equation}
    \label{eq:rip}
    \begin{aligned}
        |1-\delta_k|\|\bx\|_2^2\leq \|\bK\bx\|_2^2 \leq |1+\delta_k| \|\bx\|_2^2 \:
    \end{aligned}
\end{equation}
for all $k$-sparse vectors $\bx$, where $\delta_k$ is known as the isometry constant. 
If the RIP condition \eqref{eq:rip} holds for $\delta_{2k}<\sqrt{2}-1$, then $\bx$ can be recovered exactly using the BP program (\ref{eq:l_1_norm}) \cite{candes2008restricted}. 
Tropp developed an alternative sufficient condition, dubbed exact recovery condition (ERC), which guarantees exact recovery of the sparse vector $\bx$ using BP and OMP 
\cite{tropp2004greed}. 

As many of the aforementioned conditions 
are hard to verify in polynomial time, 
\cite{tropp2004greed} derived alternative conditions based on the notions of coherence \cite{Fuchs2004} and cumulative coherence defined in \eqref{eq:standard_coherence} and \eqref{eq:cumulative_coherence} respectively, as
\begin{equation}
    \label{eq:standard_coherence}
    \begin{aligned}
    \mu = \sup_{1\leq i<j\leq N} |\langle \bk_i,\bk_j\rangle|\:,
    \end{aligned}
\end{equation}
\begin{equation}
    \label{eq:cumulative_coherence}
    \begin{aligned}
    \mu_1\left(k\right) \eqdef \max_{\Lambda_k} \max_{i \notin 
    \Lambda_k} \sum_{j \in \Lambda_k} |\langle \bk_i,\bk_j\rangle|\:,
    \end{aligned}
\end{equation}
where $\langle.,.\rangle$ denotes the inner product and $\Lambda_k$ is a subset of atoms of size $k$. 
The author established that a sufficient condition to recover $\bx$ with sparsity level $k$ is
\begin{equation}
    \label{eq:coherence_guarantee}
    \begin{aligned}
        k < \frac{1}{2}\left(\mu^{-1}+1\right)\:.
    \end{aligned}
\end{equation}
Also, an improved condition to recover $\bx$ exactly is 
\begin{equation}
    \label{eq:cumulative_coherence_guarantee}
    \begin{aligned}
        \mu_1\left(k\right) + \mu_1\left(k-1\right) < 1
    \end{aligned}.
\end{equation}
In the presence of noise, where exact recovery is unrealizable, the literature abounds with \emph{exact support recovery} guarantees, in which it is sought to identify the exact locations of the nonzero entries of the sparse vector \cite{tropp2006just, cai2011orthogonal}.

Existing guarantees are of limited use in the context of periodic signals with sparse representations in NPDs (see Section \ref{sec:result}). This calls for alternative conditions that account for structures intrinsic to such dictionaries, which is the primary goal of this work.

\subsection{Nested periodic matrices}
\label{subsec:NPM}
In this section, we briefly review some of the important properties of nested periodic matrices (NPMs) -- the building blocks of NPDs -- which underpin period estimation via support recovery. 

A $p \times p$ NPM has the following structure \cite{tenneti2015nested}
\begin{equation}
    \label{eq:npm}
    \bA_{p} = 
    \begin{bmatrix}
    \bC_{q_1}&\bC_{q_2} \ldots \bC_{q_k}
    \end{bmatrix},
\end{equation}
where $q_i$'s are the divisors of $p$ in an increasing order, and $\bC_{q_i}$ is a $p\times \phi\left(q_i\right)$ matrix whose columns are periodic with period $q_i$, noting that $\sum_{q_i|p}\phi\left(q_i\right) = p$ \cite{hardy2008introduction}. 
The matrix $\bA_p$ has full rank. From \cite[Lemma 1]{tenneti2015nested}, any periodic $p\times 1$ vector $\by$ with period $q|p$ can be spanned by $q$ columns of $\bA_p$ that have period $q$ or a divisor of $q$. Also, from \cite[Lemma 3]{tenneti2015nested}, any linear combination $\by$ of the columns of $\bA_p$ is periodic with period exactly equal to $\lcm\{n_i\}$, where $\{n_i\}$ is the set consisting of the periods of the columns that contribute to $\by$, which is known as the LCM property. 
There are numerous NPM constructions. Next, we describe a few relevant instances, then describe how they can be used to construct NPDs.
\subsubsection{Ramanujan periodicity transform (RPT)}
The authors in \cite{tenneti2015nested} introduced the RPT matrix, which is constructed from the Ramanujan sums \cite{ramanujan1918sum},
\begin{equation}
\label{eq:ram_sum}
\begin{aligned}
c_q(n) = \sum_{\substack{k=1\\(k,q)=1}}^qe^{j2\pi kn/q}\:,
\end{aligned}
\end{equation}
where $\left(k,q\right)$ is the greatest common divisor (gcd) of $k$ and $q$. In \eqref{eq:ram_sum}, $c_q\left(n\right)$ is an all-integer sequence and periodic with period $q$. The following are few examples demonstrating one cycle of such sequences $c_q(n)$,
\begin{equation*}
\label{eq:ram_examples}
    \begin{aligned}
    c_1(n) =& \{1\},~
    c_2(n) =\{1,-1\}, ~
    c_3(n) =\{2,-1,-1\},\\
    c_4(n) =&\{2,0,-2,0\},~ 
    c_5(n) =\{4,-1,-1,-1,-1\}.
    \end{aligned}
\end{equation*} 


For any divisor $q$ of $p$, a $p \times \phi\left(q\right)$ matrix $\bC_{q}$ can be constructed as,
\begin{equation}
\label{eq:C_q}
\begin{aligned}
\bC_{q}=&
\begin{bmatrix}
\bc_{q} & \bc_{q}^{(1)} & \dots & \bc_{q}^{(\phi(q)-1)}
\end{bmatrix},
\end{aligned}
\end{equation}
where $\bc_q$ is the $p \times 1$ vectorized sequence $c_q\left(n\right)$,
\begin{equation}
 \bc_q =
 \begin{bmatrix}
 c_q\left(0\right)&c_q\left(1\right)&\ldots&c_q\left(p-1\right)
 \end{bmatrix},  
\end{equation}
and $\bc_q^{(i)}$ is the circularly downshifted version of $\bc_q$ with order $i$. For example,
\begin{equation}
 \bc_q^{\left(1\right)} = 
 \begin{bmatrix}
 c_q\left(p-1\right)&c_q\left(0\right)&\ldots&c_q\left(p-2\right)
 \end{bmatrix}.   
\end{equation}
The $p \times p$ NPM can be constructed by concatenating the submatrices $\bC_{q_i}$'s, for all $q_i|p$ as in \eqref{eq:npm}. 
Note that all the submatrices $\bC_{q_i}$ are periodic and of size $p \times \phi\left(q_i\right)$.
It was shown in \cite{vaidyanathan2014ramanujanI} that all $\phi\left(q_i\right)$ columns of the submatrices $\bC_{q_i}$ 
are linearly independent. The resulting $p \times p$ matrix $\bA_p$ is known as the RPT matrix \cite{tenneti2015nested}. As an example, 
\begin{equation}
\label{eq:npm_4}
    \begin{aligned}
\bA_4 = 
\begin{bmatrix}
    1& 1& 2&0\\
    1& -1&0&2\\
    1& 1&-2&0\\
    1& -1&0&-2
\end{bmatrix} \:.
    \end{aligned}
\end{equation}
\subsubsection{Farey dictionary}
\label{subsec:DFT}
Another example of NPMs is constructed from the discrete Fourier transform (DFT). In particular, for every $q_i|p$, we can form the $p \times \phi\left(q_i\right)$ submatrix $\bC_{q_i}$ using the signals 
\begin{equation}
\label{eq:DFT}
       b_{q_i}\left(n\right) = 
        W_{q_i}^{-kn}, \quad  \left(k,q_i\right) = 1,\quad  k = 1,\ldots,q_i,
\end{equation}
where $W_{q_i} = e^{-j2\pi/q_i}$ is the complex exponential. Note that there are exactly $\phi\left(q_i\right)$ values for $k$ that obey the condition $\left(k,q_i\right) = 1$. Therefore, there are $\phi\left(q_i\right)$ periodic signals with period $q_i$, which together form the submatrix $\bC_{q_i}$.

\subsection{Nested periodic dictionaries (NPDs)}
\label{subsec:NPD}
A method to construct an NPD that contains exactly $\phi\left(p\right)$ atoms spanning the corresponding nested periodic subspace is described in \cite{tenneti2015nested}. The main idea is to first build the $p \times p$ NPMs, for every $p \in \mathbb{P} := \{1,\ldots,P_{\max}\}$, where $P_{\max}$ is the largest possible period, and keep the last $\phi\left(p\right)$ columns yielding submatrices $\bC_p$'s, for all $p \in \mathbb{P}$. Then, the NPD $\bK$ is constructed by concatenating matrices $\bR_p$ obtained by periodically extending the matrices $\bC_p$ to length $L$ as
\begin{equation}
    \label{eq:NPD}
    \bK = \begin{bmatrix}
    \bR_1&\bR_2&\ldots&\bR_{P_{\max}}
    \end{bmatrix}\:.
\end{equation}

\subsection{Ramanujan subspaces}\label{subsec:periodic_subspaces}
Let $\mathcal{V}_p$ denote the periodic subspace that contains all signals with period $p$. Periodically extending the columns of $\bA_p$ yields a basis that spans $\mathcal{V}_p$. In addition, a subset of the columns spans $\mathcal{V}_{q_i}$, for all $q_i|p$. The subspace $\mathcal{V}_p$ is the direct sum of the subspaces spanned by the columns of submatrices $\bC_{q_i}$, where $q_i|p$. The Ramanujan subspace $\mathcal{S}_q$ is the column space of the periodically extended version of the matrix $\bC_q$ in \eqref{eq:C_q} \cite{vaidyanathan2014ramanujanI}. Further, it was shown that 
$\bigoplus_{q_i|p} \mathcal{S}_{q_i} = \mathcal{V}_p$.

Two compelling features of Ramanujan subspaces are their (i) Euler structure: the dimension of $\mathcal{S}_q$ is $\phi\left(q\right)$ \cite{tenneti2015nested,vaidyanathan2014ramanujanI,tenneti2016unified}; 
    (ii) LCM property: the period of mixtures in the Ramanujan subspaces of the form \eqref{eq:mixture_definition} is exactly equal to the LCM of its constituent hidden periods and not a divisor of the LCM \cite{tenneti2016unified}. 

Ramanujan subspaces are special \emph{orthogonal} nested periodic subspaces (NPSs) (because of the orthogonality of the subspaces $\mathcal{S}_{q_i}$). An NPS is a broader term used in \cite{tenneti2016unified} to refer to the span of the $\bC_{q}$ matrices of the NPMs, and such subspaces are generally non-orthogonal.

\subsection{Period estimation via support recovery} \label{subsec:sup_recovery_period_estimation}
Using sparse recovery methods like OMP and BP, it has been empirically demonstrated that the support set of a sparse representation of a periodic mixture with hidden periods $p_i \in \mathbb{P} := \{1,\ldots,P_{\max}\}$ can be recovered from an NPD of size $L \times N$. This allows for estimation of the hidden periods using the LCM property. In general, the recovered support set is similar to the case where $L$ is infinitely large, leading to successful period estimation. However, this is not always the case, as illustrated in Fig. \ref{fig:bp_examples}. For instance, the sparse representation of a periodic signal with period 20 and $L=200$ reveals the period as the LCM of $1, 2, 4, 5, 10, 20$ (see Fig. \ref{subfig:coef_200}). In contrast, the true period cannot be estimated from the sparse representation of a periodic signal with period 20 and $L=80$ (see Fig. \ref{subfig:coef_80}). It is important to note that there is currently no theoretical condition to guarantee exact recovery of the support set of the sparse vector and successful period estimation using the LCM property.
\begin{figure}
    \centering
    \begin{subfigure}[b]{0.46\textwidth}
    \centering
        \includegraphics[width = \textwidth,height = 6cm]{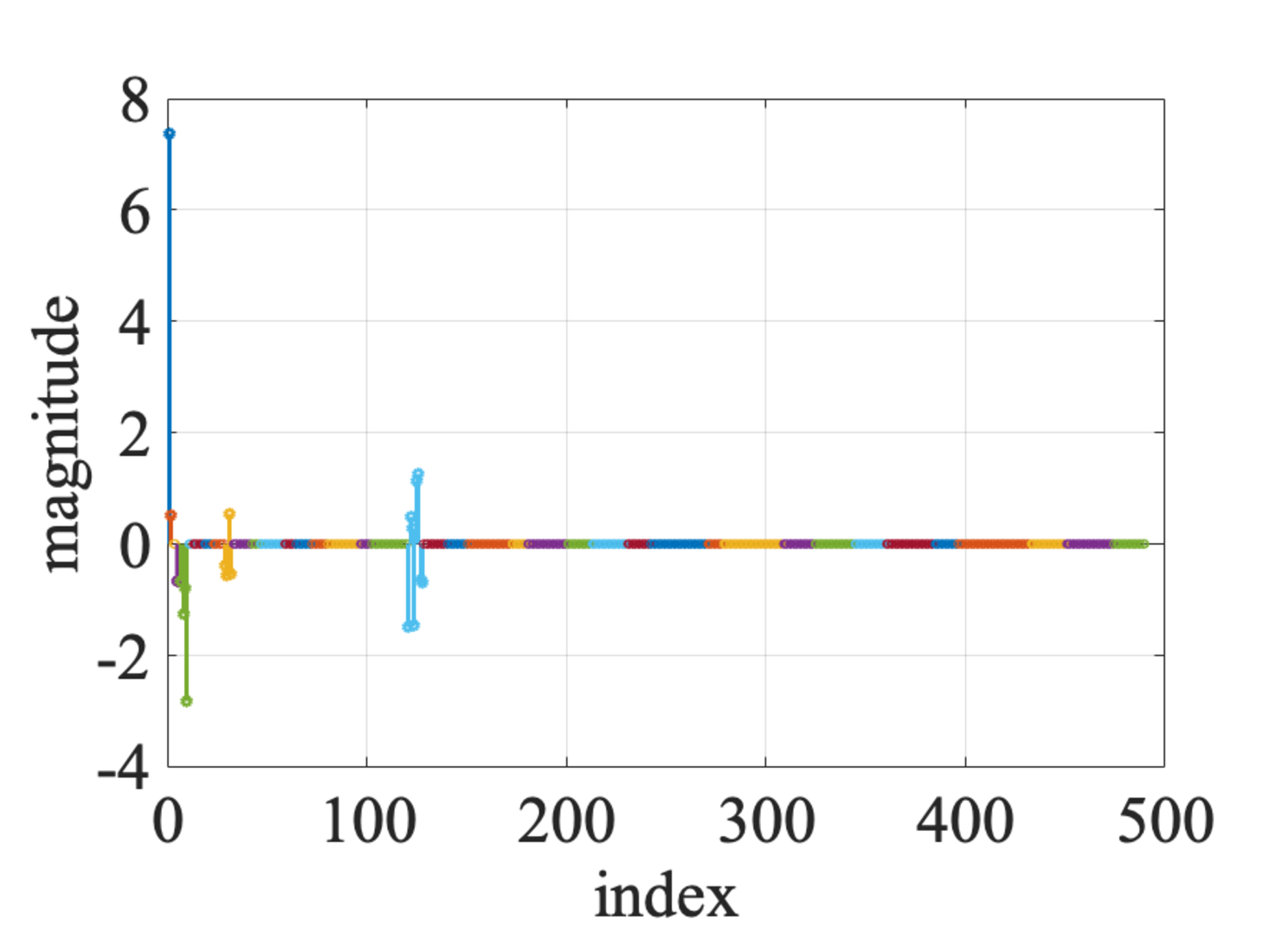}
        \caption{}
        \label{subfig:coef_200}
    \end{subfigure}
    \hfill
    \begin{subfigure}[b]{0.46\textwidth}
    \centering
        \includegraphics[width = \textwidth,height = 6cm]{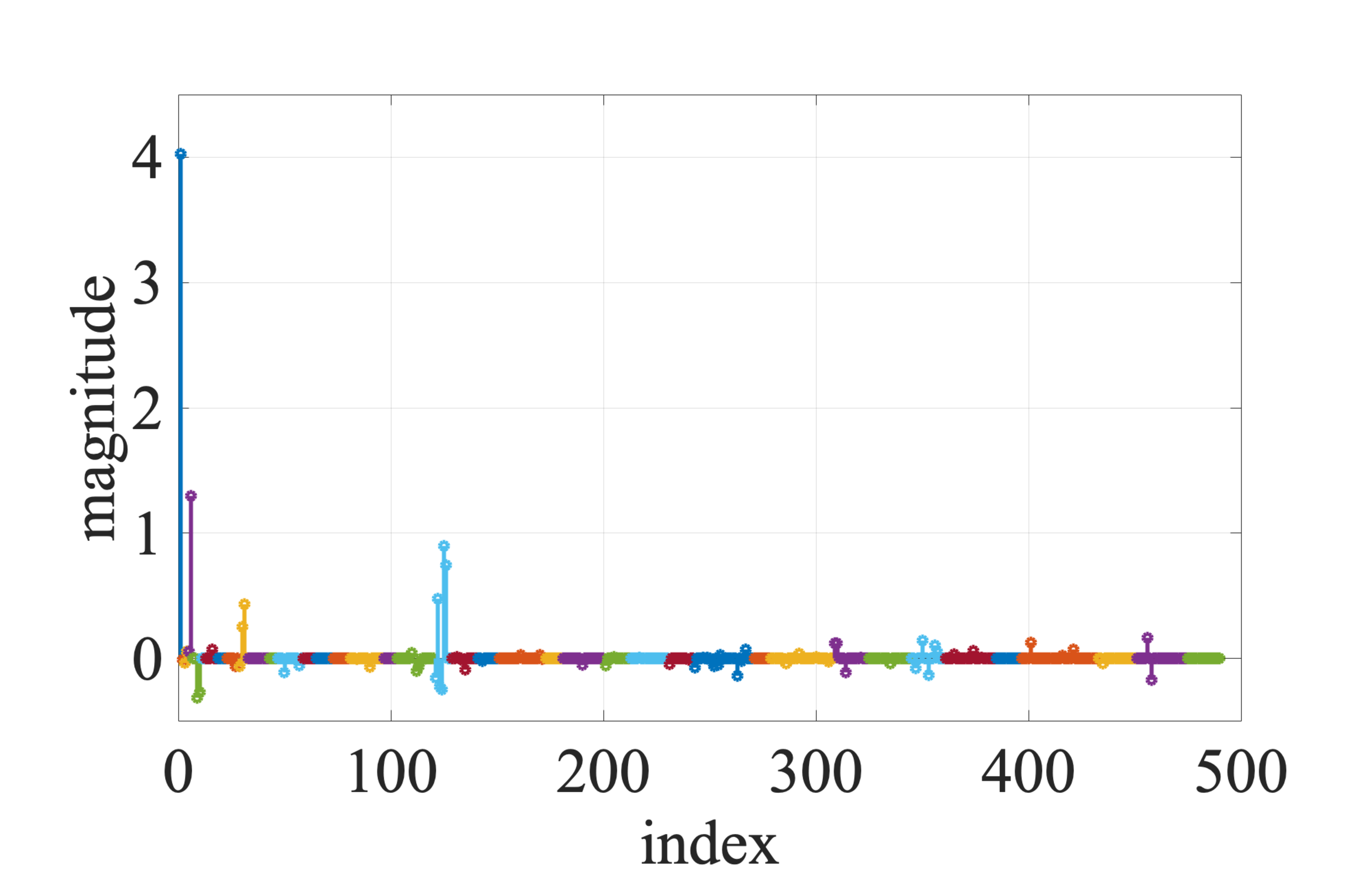}
    \caption{}
    \label{subfig:coef_80}
    \end{subfigure}
    \caption{The nonzero coefficients of the sparse representations of the periodic signals with period 20 (a) L = 200 and (b) L = 80. Here, each group of nonzero coefficients illustrated with a distinct color corresponds to the same submatrix $\bR_p$.}
    \label{fig:bp_examples_coef}
\end{figure}

\section{Problem Setup and Main Results}
\label{sec:contribution}
In a matrix form, a periodic signal of length $L$ can be modeled as
\begin{equation}
    \label{eq:NPD_model}
    \begin{aligned}
    \by = \bK\bx
    \end{aligned}\:,
\end{equation}
where $\bK \in \mathbb{C}^{L\times N}$ is the NPD with $L<N$, and $\by$ is the vector form of $y\left(n\right)$ and  $\bx \in \mathbb{C}^{N}$ is its representation in $\bK$. In the presence of noise, the periodic signal model can be modified to
\begin{equation}
    \label{eq:NPD_model_noisy}
    \begin{aligned}
    \tilde{\by} = \bK\bx +\bw
    \end{aligned},
\end{equation}
where 
$\bw \in \mathbb{R}^L$ denotes the noise vector.
A periodic signal $\by$ composed of a mixture of hidden periods lives in the union of some of the subspaces spanned by the submatrices of the NPDs \cite{tenneti2015nested}. Therefore, estimating the hidden periods of a mixture is equivalent to identifying the correct column support of its representation in an NPD. Given an underdetermined system of equations as the one in \eqref{eq:NPD_model}, i.e., when the number of measurements (rows) $L$ is smaller than the number of atoms (columns) $N$, the support set of the sparse vector $\bx$ in (\ref{eq:NPD_model}) can be identified using a sparse recovery framework \cite{tenneti2018minimum,vaidyanathan2014farey}.
The performance of such techniques depend on multiple features of NPDs, including the functions used to construct these dictionaries and the datalength $L$ \cite{tenneti2018minimum}. The authors in \cite{tenneti2018minimum} established that $L\geq 2P_{\max}$ is both necessary and sufficient to recover the exact support set of the sparse vector of a periodic signal (single period) using the program in \eqref{eq:l_0_norm}, however, as discussed in Section \ref{sec:background} this is an NP-hard problem and is of limited practical use.

We start off with the noise free case where the discrete-periodic signal is modeled as in \eqref{eq:NPD_model}. In general, we consider $y$ to be a periodic mixture that consists of $m$ periodic signals as expressed in \eqref{eq:mixture_definition}. First, we introduce some definitions before stating the main results. All proofs are deferred to the appendix.
\begin{definition}[Support set]
\label{def:support_set}
We define the index set $I_p = \{\sum_{j=1}^{p-1} \phi\left(j\right)+1,\ldots,\sum_{j=1}^{p}\phi\left(j\right)\}$, which consists of the indices of the $\phi\left(p\right)$ atoms of an NPD $\bK$ that have period $p$. 
\end{definition}

\begin{example}
For $p = 5$, we have $\phi\left(5\right) = 4$, and $\sum_{j=1}^{4} \phi\left(j\right) = 6$, therefore, $I_5 = \{7,8,9,10\}$ as shown in Fig. \ref{fig:RPT_example}.
\end{example}
\begin{figure}
    \centering
    \includegraphics[width = 0.5\linewidth,height = 4cm]{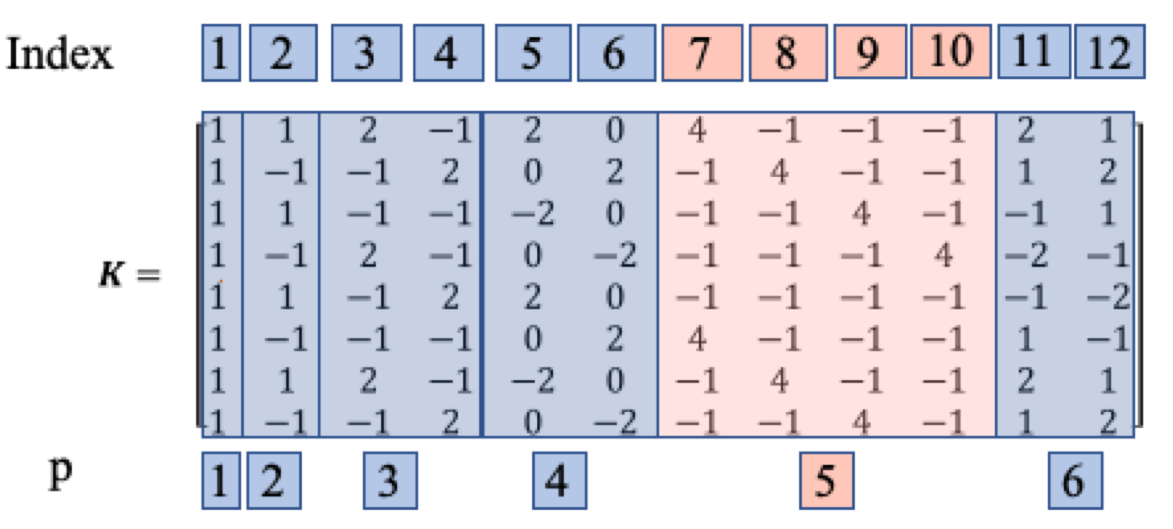}
    \caption{RPT dictionary with $P_{\max} = 6$ and $L = 8$. The periodically extended version of submatrix $\bC_5$ and its corresponding indices in $I_5$ are shown in a distinct color. The columns of the dictionary are not normalized.}
    \label{fig:RPT_example}
\end{figure}
Definition \ref{def:support_set} determines the support set corresponding to the atoms that span the nested periodic subspace of period $p$. Our next definition generalizes this notion to the union of the support sets of a periodic mixture. 
\begin{definition} [Union of support sets]
\label{def:union_support_set}
Let $\mathcal{T} \in 2^\mathbb{P}$ be a non-empty set that contains all the hidden periods of a periodic mixture and
  $Q_p = \{q:  q|p\}$ be the set with all divisors of $p$. Then, $D_{\mathcal{T}} = \bigcup_{p \in \mathcal{T}} Q_p$ is the set of all divisors of the elements in $\mathcal{T}$ and $S_{\mathcal{T}} := \bigcup_{q \in D_{\mathcal{T}}} I_q$ is the set containing the support of the 
  periodic mixture (obtained as the union of the index sets $I_q$ over all the divisors in $ D_{\mathcal{T}}$). For the special case where $m = 1$, we use the notation $S_p$, $p \in \mathbb{P}$, for the support set, defined as the union of the sets $I_q$ for all the divisors $q\in Q_p$, i.e., $S_{p} := \bigcup_{q \in Q_p} I_q$. 
\end{definition}

\begin{example} \label{ex:union_suppport_set}
Assume $P_{\max} = 8$, $\mathcal{T}_1 = \{3,5\} \in 2^\mathbb{P}$. Then $D_{\mathcal{T}_1} = \{1,3,5\}$ and $S_{\mathcal{T}_1} = \{1,3,4,7,8,9,10\}$. For $\mathcal{T}_2 = \{3,4\}$, we have $D_{\mathcal{T}_2} = \{1,2,3,4\}$ and $S_{\mathcal{T}_2} = \{1,2,3,4,5,6\}$.
\end{example}

We remark that a sparse mixture of $m$ elements could induce different sparsity levels as highlighted in Example \ref{ex:union_suppport_set}. So, henceforth we will focus on a set $\mathbb{Q}_k\left(m\right)$ with bounded sparsity defined next.
\begin{definition} \label{def:Qkm}
Denote by $\mathbb{Q}\left(m\right)$ the collection of all subsets of $2^\mathbb{P}$ with $m$ elements that are not divisors of one another, i.e., 
\[
\mathbb{Q}\left(m\right):=\{\mathcal{T}\in 2^\mathbb{P}: |\mathcal{T}|=m ~\wedge ~ p_i \not\vert p_j, \forall p_i, p_j \in\mathcal{T} \}.
\] 
Further, define 
$\mathbb{Q}_k(m):=\{\mathcal{T}\in\mathbb{Q}\left(m\right): |S_{\mathcal{T}}|\leq k\}$. 
\end{definition}

\subsection{Exact sparse recovery guarantees in the noise-free case}
\label{subsec:method_noiseless}
In this section, we establish exact recovery guarantees for periodic signals. Conditions such as \eqref{eq:coherence_guarantee} and \eqref{eq:cumulative_coherence_guarantee} ignore the structure of the NPDs, thus do not yield meaningful recovery bounds in this context, as will be shown in Section \ref{sec:result}. 
Therefore, we leverage the properties of NPDs, namely the Euler structure and the LCM property, to provide new bounds tailored to NPDs. We seek improved conditions in the spirit of \eqref{eq:coherence_guarantee} and \eqref{eq:cumulative_coherence_guarantee} that provide achievable bounds for large sparsity level $k$, and are easy to verify. To this end, we introduce the notion of nested periodic inter-coherence (NPI) and nested periodic intra-coherence (NPA) that restrict the coherence calculation to the set $\mathbb{Q}_k\left(m\right)$.

\begin{definition}
\label{def:npi_coh}
Given an $L\times N$ NPD with $N = \sum_{p = 1}^{P_{\max}}\phi\left(p\right)$, we define the nested periodic inter-coherence and 
intra-coherence in \eqref{eq:npi_coherence} and \eqref{eq:npa_coherence}, respectively, as,   
\begin{equation}
    \label{eq:npi_coherence}
    \zeta_{k,m} \eqdef  \max_{\mathcal{T}\in \mathbb{Q}_k\left(m\right)} \max_{\substack{i \in S_{\mathcal{T}}^c}} \sum_{\substack{j \in S_{\mathcal{T}}}} |\langle \bk_i,\bk_j\rangle|\:,
\end{equation}
and
\begin{equation}
    \label{eq:npa_coherence}
    \nu_{k,m} \eqdef  \max_{\mathcal{T}\in \mathbb{Q}_k\left(m\right)} \max_{\substack{i \in S_{\mathcal{T}}}} \sum_{\substack{j \in S_{\mathcal{T}}\\i\neq j}} |\langle \bk_i,\bk_j\rangle|\:.
\end{equation}

\end{definition}

\textcolor{black}{The schematics in Fig. \ref{fig:coherence_ex2}.a and \ref{fig:coherence_ex2}.b illustrate the calculation of the NPI restricted to  $\mathcal{T} = \{3,4\} \in \mathbb{Q}_6\left(2\right)$. The atoms in $S^c_{\mathcal{T}}$ and $S_{\mathcal{T}}$ are shown in green and blue boxes, respectively. The blue squares in Fig. \ref{fig:coherence_ex2}.b mark the inner products of the atoms corresponding to sets $S_{\mathcal{T}}$ and $S_{\mathcal{T}}^c$ for $\mathcal{T} = \{3,4\}$. The red rectangles designate the summands in the summation over $j\in S_{\mathcal{T}}$. The maximum of these summations is the inner maximization in \eqref{eq:npi_coherence}. Finally, by maximizing over all sets in $\mathbb{Q}_6\left(2\right)$, we obtain $\zeta_{6,2}$ (the outer maximization in \eqref{eq:npi_coherence}).}
Based on these new  definitions, we can obtain a tight upper bound on an ERC-type condition when restricted to NPDs and periodic mixtures for sets in $\mathbb{Q}_k\left(m\right)$, and in turn the new recovery condition in Theorem \ref{thr:guarantee_upperbound_mixture_strong} for corresponding periodic mixtures.

\begin{theorem}[Recovery condition for mixtures in $\mathbb{Q}_k\left(m\right)$]
\label{thr:guarantee_upperbound_mixture_strong}
Suppose that $\bK$ is an NPD of size $L\times N$, $\mathcal{T} \in \mathbb{Q}_k\left(m\right)$ contains the hidden periods of a periodic mixture where $\mathbb{Q}_k\left(m\right)$ is as in Definition \ref{def:Qkm}, and $\zeta_{k,m}$ and $\nu_{k,m}$ are the NPI and NPA in \eqref{eq:npi_coherence} and \eqref{eq:npa_coherence}, respectively.
If the following holds for the given number of hidden periods $m$ and sparsity level $k$,
\begin{equation}
    \label{eq:guarantee_NPD_strong}
    \zeta_{k,m} +  \nu_{k,m} < 1 \:,
\end{equation}
then one can identify all the hidden periods of the periodic mixture by recovering the sparse vector $\bx$ using BP and OMP.
\end{theorem}

The condition in \eqref{eq:guarantee_NPD_strong} is a tight upper bound for the ERC restricted to NPDs, and is computationally more efficient than the ERC. However, since it is restricted to mixtures with hidden periods in $\mathbb{Q}_k\left(m\right)$, we seek an alternative bound that is simultaneously easy to compute and not restricted to a certain sparsity level $k$ and number of hidden periods $m$.  
To this end, we introduce the notions of restricted inter-coherence and intra-coherence.
\begin{definition} \label{def:restricted_cumulative_coh}
We define the restricted inter-coherence and restricted intra-coherence of an NPD $\bK$ for period $p$ in \eqref{eq:restricted_inter_coh_special} and \eqref{eq:restricted_intra_coh_special}, respectively,
\begin{equation}
    \label{eq:restricted_inter_coh_special}
    \begin{aligned}
    \zeta_p \eqdef \max_{i \in S_p^c} \sum_{j \in S_p} |\langle \bk_i,\bk_j\rangle|\:,
    \end{aligned}
 \end{equation}
 and
\begin{equation}
    \label{eq:restricted_intra_coh_special}
    \nu_p \eqdef  \max_{\substack{i \in S_{p}}} \sum_{\substack{j \in S_{p}\\i\neq j}} |\langle \bk_i,\bk_j\rangle|\:.
\end{equation}
where the set $S_p$ is given in Definition \ref{def:union_support_set}.
\end{definition}
Next, we establish a sufficient condition for recovering the exact sparse representations of periodic mixtures in an NPD using the notion of restricted coherence.

\begin{theorem}[Computationally efficient alternative bound]
\label{thr:single_2_mixture}
Suppose $\bK$ is an NPD and the set $\mathcal{T}$ contains the hidden periods of a periodic mixture. Let $\zeta_p$ and $\nu_p$ be the restricted inter-coherence and restricted intra-coherence of $\bK$ for period $p \in \mathbb{P}$ as defined in \eqref{eq:restricted_inter_coh_special} and \eqref{eq:restricted_intra_coh_special}, respectively. If the condition in \eqref{eq:single_2_mixture} holds, then BP and OMP can recover the exact sparse vector $\bx$ in \eqref{eq:NPD_model} that represents the periodic mixture.
\begin{equation}
    \label{eq:single_2_mixture}
    2\left(\sum_{p_j\in \mathcal{T}} \zeta_{p_j}\right) +  \max_{p_j\in\mathcal{T}}\nu_{p_j} - \min_{p_j\in\mathcal{T}} \zeta_{p_j} < 1 \:.
\end{equation}
\end{theorem}

The result of Theorem \ref{thr:single_2_mixture} reduces to the following for the special case when $m = 1$,
\begin{equation}
    \label{eq:guarantee_restricted_coh_NPD}
    \begin{aligned}
        \zeta_p + \nu_p < 1 \:.
    \end{aligned}
\end{equation}
The condition of Theorem \ref{thr:single_2_mixture} in \eqref{eq:single_2_mixture} is less tight than the previous bound in \eqref{eq:guarantee_NPD_strong} since Definition \ref{def:restricted_cumulative_coh} ignores the common divisors between the hidden periods in the calculations -- versus Definition \ref{def:npi_coh}, which accounts for the union of divisors. However, 
condition \eqref{eq:single_2_mixture} significantly improves the amount of computations and can be more easily verified than \eqref{eq:guarantee_NPD_strong}. Another advantage is that it is not restricted to a fixed $k$ and $m$, thus is applicable to arbitrary periodic mixtures.

We can also obtain a comparable bound to that of Theorem \ref{thr:guarantee_upperbound_mixture_strong} for signals defined by the set $\mathbb{Q}_k\left(m\right)$ as a consequence of Theorem \ref{thr:single_2_mixture}. 
\begin{corollary}
\label{cor:single_2_mixture}
Let $\mathcal{T} \in \mathbb{Q}_k\left(m\right)$ be a set that contains the hidden periods of a periodic mixture. If
\begin{equation}
   \label{eq:single_2_mixture_Qkm}
    \max_{\mathcal{T}\in \mathbb{Q}_k\left(m\right)}    \left(2\left(\sum_{p_j\in \mathcal{T}} \zeta_{p_j}\right) + \max_{p_j\in\mathcal{T}}\nu_{p_j} - \min_{p_j\in\mathcal{T}} \zeta_{p_j}\right) < 1 \:,
\end{equation}
then BP and OMP can recover the representation of any periodic mixture of $m$ periodic signals admitting a $k$-sparse representation in an NPD.
\end{corollary}

In Appendix \ref{App:General_ERC_NPD_mixture}, we discuss the ERC as it pertains to NPDs. It is worth noting that our new definitions of NPA and NPI, as well as the restricted inter-coherence and restricted intra-coherence, allow us to establish tight upper bounds on the ERC. While it is possible to demonstrate that the ERC holds for a larger sparsity level $k$ by evaluating all elements in the set $\mathbb{Q}_k\left(m\right)$ for some values of $k$, this approach is computationally expensive. Additionally, existing bounds based on standard coherence are too loose to provide reliable sparse recovery guarantees for NPDs. However, the new results presented in Theorem \ref{thr:guarantee_upperbound_mixture_strong} and Corollary \ref{cor:single_2_mixture} offer more computationally efficient and tighter bounds on the ERC. In Section \ref{sec:result}, we provide information on the necessary run times to verify these conditions. It is important to note that while the ERC guarantees exact sparse recovery, it does not reveal much information about the dictionaries beyond this. In contrast, coherence is an interpretable statistic that can offer more insight into the dictionaries. Therefore, we introduce the NPA and NPI, which are tailored for NPDs and based on the notion of coherence. While outside the scope of this work, these metrics provide additional insight into the structure of NPDs beyond their ability to support sparse recovery.

\subsection{Support recovery guarantees in the presence of noise}
\label{subsec:method_noisy_case}
We turn our attention in this section to the model in \eqref{eq:NPD_model_noisy} in which the observation signal $\tilde{\by}$ is a periodic signal $\by$ \eqref{eq:NPD_model} contaminated with noise $\bw$. We limit the NPDs to those with real values, i.e., $\bK \in \mathbb{R}^{L \times N}$. In the context of period estimation, it is important to recover the exact support (not a subset), otherwise we could end up with an incorrect estimate of the period. 
Following \cite{cai2011orthogonal}, we focus on the OMP algorithm to recover the exact support of the periodic signal, where we use $\br_{i}$ to denote the residual at iteration $i$ of the OMP algorithm. We consider two cases for the additive noise $\bw$: the $\ell_2$-norm bounded noise, i.e., $\|\bw\|_2\leq \epsilon$, and Gaussian noise, i.e., $\bw \sim \mathcal{N}\left(\mathbf{0},\sigma^2\bI_L\right)$.

\subsubsection{Bounded noise}
In this case, the $\ell_2$-norm of the noise vector in \eqref{eq:NPD_model_noisy} is bounded, i.e., $\|\bw\|_2\leq \epsilon$, for some small $\epsilon > 0$. Following \cite{cai2011orthogonal}, we can assert the following result. 
\begin{theorem}
\label{thr:noise_bounded_general_zeta_nu}
Consider the model in \eqref{eq:NPD_model_noisy}, in which $\tilde{\by}$ is a periodic mixture of $m$ periodic signals with hidden periods in $\mathcal{T} \in \mathbb{Q}_k\left(m\right)$ and $\|\bw\|_2\leq \epsilon$. 
If \eqref{eq:guarantee_NPD_strong} holds and all the nonzero coefficients $x_i$ satisfy
\begin{equation}
    \label{eq:omp_npd_noise_bnd_zeta_nu}
|x_i|\geq \frac{2\epsilon}{\left(1-\zeta_{k,m}-\nu_{k,m}\right)\left(1-\nu_{k,m}\right)},
\end{equation}
where $\zeta_{k,m}$ and $\nu_{k,m}$ are as in \eqref{eq:npi_coherence} and \eqref{eq:npa_coherence}, respectively, then the OMP algorithm with the stopping rule $\|\br_i\|_2\leq \epsilon$ recovers the exact support of $\tilde{\by}$.
\end{theorem}
In addition, we establish another sufficient condition for the bounded noise case that leverages the alternative bound in Theorem \ref{thr:single_2_mixture}.
\begin{theorem}
\label{thr:noise_bounded_general_s2m}
Suppose $\tilde{\by}$ is a periodic mixture of $m$ hidden periods in $\mathcal{T}$ and $\|\bw\|_2\leq \epsilon$. If \eqref{eq:single_2_mixture} holds and all the nonzero coefficients $x_i$ satisfy
\begin{equation}
\begin{aligned}
\label{eq:noise_bounded_s2m}
        &|x_i|\geq\\
    &\frac{2\epsilon}{\left(1-2\sum_{p_j\in \mathcal{T}} \zeta_{p_j} - \hat{\nu}_{p} + \check{\zeta}_p\right)\left(1-\sum_{p_j\in \mathcal{T}} \zeta_{p_j}-\hat{\nu}_p + \check{\zeta}_p\right)},
    \end{aligned}
\end{equation}
where $\hat{\nu}_{p} = \max_{p_j\in\mathcal{T}}\nu_{p_j}$ and $\check{\zeta}_{p} = \min_{p_j\in\mathcal{T}} \zeta_{p_j}$, and $\zeta_p$ and $\nu_p$ are as in \eqref{eq:restricted_inter_coh_special} and \eqref{eq:restricted_intra_coh_special}, respectively, then the OMP algorithm with the stopping rule $\|\br_i\|_2\leq \epsilon$ recovers the exact support of $\tilde{\by}$.
\end{theorem}
\subsubsection{Gaussian noise}
Suppose $\bw \sim \mathcal{N}\left(\mathbf{0},\sigma^2\bI_L\right)$.
\textcolor{black}{We extend the result of Theorem \ref{thr:noise_bounded_general_s2m} to the Gaussian case.
\begin{theorem}
 \label{thr:noise_gaussian_s2m}
Suppose $\bw \sim \mathcal{N}\left(\mathbf{0},\sigma^2\bI_L\right)$, and $\tilde{\by}$ is a periodic mixture of $m$ periodic signals with its hidden periods in $\mathcal{T}$. 
If \eqref{eq:single_2_mixture} holds and all the nonzero coefficients $x_i$ satisfy
\begin{equation}
\begin{aligned}
    \label{eq:gaussian_coeff_s2m}
    &|x_i| \geq\\ 
    &\frac{2\sigma \sqrt{L + 2\sqrt{L\log L}}}{\left(1-2\sum_{p_j\in \mathcal{T}} \zeta_{p_j} - \hat{\nu}_{p} + \check{\zeta}_p\right)\left(1-\sum_{p_j\in \mathcal{T}} \zeta_{p_j}-\hat{\nu}_p + \check{\zeta}_p\right)}\:,
    \end{aligned}
\end{equation}
then the OMP algorithm with the stopping rule $\|\br_i\|_{2}\leq \sigma \sqrt{L + 2\sqrt{L \log L}}$ recovers the exact support with probability at least $1-1/L$.
\end{theorem}

We can establish another result for the Gaussian case similar to Theorem \ref{thr:noise_bounded_general_zeta_nu} based on the condition in \eqref{eq:guarantee_NPD_strong}. This result is omitted for brevity.}

\subsection{Extension to refined conditions in the noise-free case}
Our numerical results in section \ref{sec:result} reveal that Theorem \ref{thr:guarantee_upperbound_mixture_strong} offers significant improvement over the condition in  \eqref{eq:cumulative_coherence_guarantee}, in the sense that it provides an achievable bound for larger sparsity levels $k$. Nevertheless, the condition in \eqref{eq:guarantee_NPD_strong} based on the NPI and NPA is conservative in two ways. First, the summation in \eqref{eq:npi_coherence} and \eqref{eq:npa_coherence} is over all $|S_\mathcal{T}|$ values, whereas the support set $\tilde{S}\subseteq S_\mathcal{T}$, so the upper bound in \eqref{eq:guarantee_NPD_strong} overestimates the coherence. 
Second, the outer maximization in Definition \ref{def:npi_coh} is defined over all sets in $\mathbb{Q}_k\left(m\right)$. As such, the sets $\mathcal{T} \in \mathbb{Q}_k\left(m\right)$ with support sets $S_{\mathcal{T}}$ of cardinality smaller than $k$ are typically dominated by other sets in $\mathbb{Q}_k\left(m\right)$. 
In this section, we derive a more refined bound by accounting for the actual sparsity level $s \leq k$, and treating $k$ as the parameter characterizing the set $\mathbb{Q}_k\left(m\right)$. 
First, we introduce the cumulative nested periodic inter-coherence (CNPI) and cumulative nested periodic intra-coherence (CNPA) for set $\mathbb{Q}_k\left(m\right)$. These notions account explicitly for sets $S_{\mathcal{T}}$ of cardinality $s$ smaller than $k$. As a result, we are able to derive a refined condition 
that resolves the foregoing limitations and gives more insight into the phase transitions in the $s$\textendash $k$-plane.
\begin{definition}
\label{def:npi_cc}
Given an $L\times N$ NPD with $N = \sum_{p = 1}^{P_{\max}}\phi\left(p\right)$, we define the cumulative nested periodic inter-coherence and cumulative nested periodic intra-coherence in \eqref{eq:npi_cc} and \eqref{eq:npa_cc}, respectively, as
\begin{equation}
    \label{eq:npi_cc}
    \zeta_{k,m}\left(s\right) \eqdef  \max_{\mathcal{T}\in \mathbb{Q}_k\left(m\right)} \max_{\substack{i \in S_{\mathcal{T}}^c}} \max_{\Lambda_s} \sum_{\substack{j \in \Lambda_s}} |\langle \bk_i,\bk_j\rangle|\:,
\end{equation}
and
\begin{equation}
    \label{eq:npa_cc}
    \nu_{k,m}\left(s\right) \eqdef  \max_{\mathcal{T}\in \mathbb{Q}_k\left(m\right)} \max_{\substack{i \in S_{\mathcal{T}}}} \max_{\Lambda_s} \sum_{\substack{j \in \Lambda_s\\i\neq j}} |\langle \bk_i,\bk_j\rangle|\:.
\end{equation}
where $\Lambda_s \subseteq S_{\mathcal{T}}$ denotes a subset of at most $s\leq k$ nonzero coefficients.
\end{definition}

Fig. \ref{fig:coherence_ex2}.c contrasts the NPI with the CNPI. As an example, for $\mathcal{T} = \{3,4\}$, $k = 6$ and  $s = 3$, instead of summing over all $j \in S_{\mathcal{T}}$, for the latter we find the summation over the $s$ largest values. It is pertinent to note that the sparse recovery algorithms are oblivious to the values of $s$ and $k$. Based on these new definitions, we assert the following theorem.
\begin{figure*}
    \centering
        \includegraphics[width = \textwidth,height = 9cm]{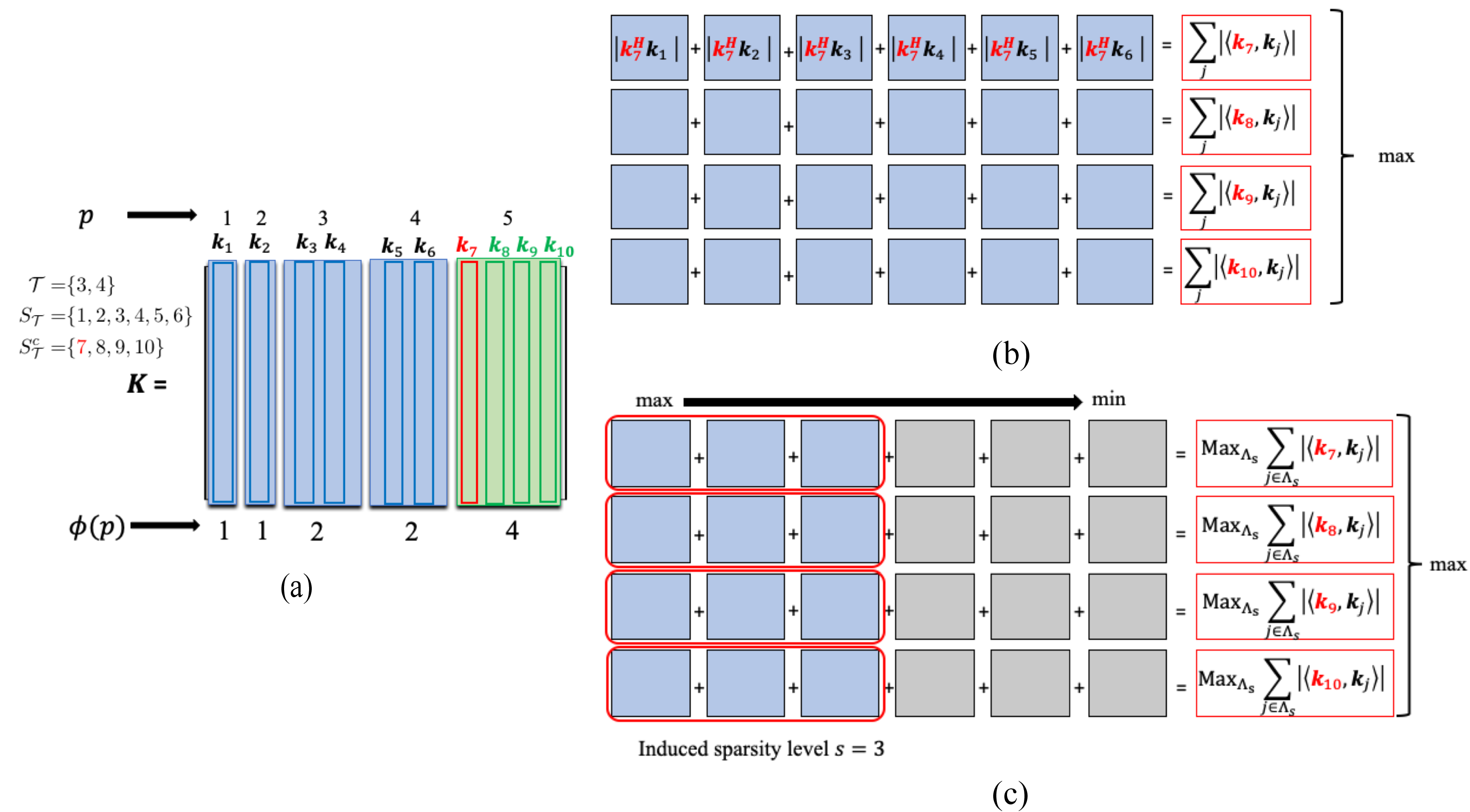}
    \caption{Calculation of the NPI in \eqref{eq:npi_coherence}. (a) For $\mathcal{T} = \{3,4\} \in \mathbb{Q}_6\left(2\right)$, the atoms in $S^c_{\mathcal{T}}$ and $S_{\mathcal{T}}$ are shown in green and blue boxes, respectively. The inner summation in \eqref{eq:npi_coherence} is over the absolute value of the inner products between one atom from $S^c_{\mathcal{T}}$ (e.g., $i = 7$ marked with a red rectangle) and all the atoms in $S_{\mathcal{T}}$. (b) The inner products of the atoms in $S_\mathcal{T}$ and $S_\mathcal{T}^c$ are shown in blue color. Each row corresponds to an atom from $S_\mathcal{T}^c$, e.g., the first row shows the summands (inner products) in the summation over $j\in S_{\mathcal{T}}$ for $i = 7$. The inner maximization in \eqref{eq:npi_coherence} selects the maximizing row.
    (c) The difference between the NPI and CNPI, for $\mathcal{T} = \{3,4\}$ and $s = 3$: instead of summing over all $j \in S_{\mathcal{T}}$, we sum over the largest $s$ values.}
    \label{fig:coherence_ex2}
\end{figure*}

\begin{theorem}[Refined recovery bound]
\label{thr:guarantee_mixture_sparse_based}
Suppose that $\bK$ is an NPD of size $L \times N$, and $\zeta_{k,m}\left(s\right)$ and $\nu_{k,m}\left(s\right)$ are the CNPI and CNPA in \eqref{eq:npi_cc} and \eqref{eq:npa_cc}, respectively. Let $\mathcal{T}\in \mathbb{Q}_k\left(m\right)$ contain the hidden periods of a periodic mixture that admits a sparse representation with sparsity level $s\leq k$ in the NPD. If the condition in \eqref{eq:guarantee_NPD_mixture_sparse} holds, then BP and OMP can identify the hidden periods of the periodic mixture by recovering its sparse representation in the NPD.
\begin{equation}
    \label{eq:guarantee_NPD_mixture_sparse}
    \zeta_{k,m}\left(s\right) + \nu_{k,m}\left(s-1\right) < 1 \:.
\end{equation}
\end{theorem}

Theorem \ref{thr:guarantee_mixture_sparse_based} rests on the notion of cumulative nested periodic coherence introduced in Definition \ref{def:npi_cc} and provides a more relaxed condition when the sparsity level is $s$ for some $s \leq k$.



\section{Numerical Results} \label{sec:result}
\subsection{Exact support recovery condition for the noise-free case} \label{subsec:result_noise_free}
 In this section, we investigate the implications of the derived conditions for periodic mixtures. 
We compute an exact recovery condition in the spirit of ERC \cite{tropp2004greed} (stated in the appendix) which serves as a baseline to gauge the tightness of the derived bounds.  While computing an ERC requires combinatorial search over support sets, the structure of NPDs allows us to restrict the search to sets $S_{\mathcal{T}}$ for every $\mathcal{T} \in \mathbb{Q}_k\left(m\right)$ given in Definition \ref{def:Qkm}.   
We refer to this baseline in our figures as $M_k\left(m\right)$. 
Fig. \ref{fig:erc_mixture} shows the results for $m = 1$ and $m = 2$, $P_{\max} = 100$ and $L = 1915$ for the Farey and RPT dictionary.
\begin{figure}
    \centering
    \begin{subfigure}[b]{0.46\textwidth}
    \centering
        \includegraphics[width = \textwidth,height = 6cm]{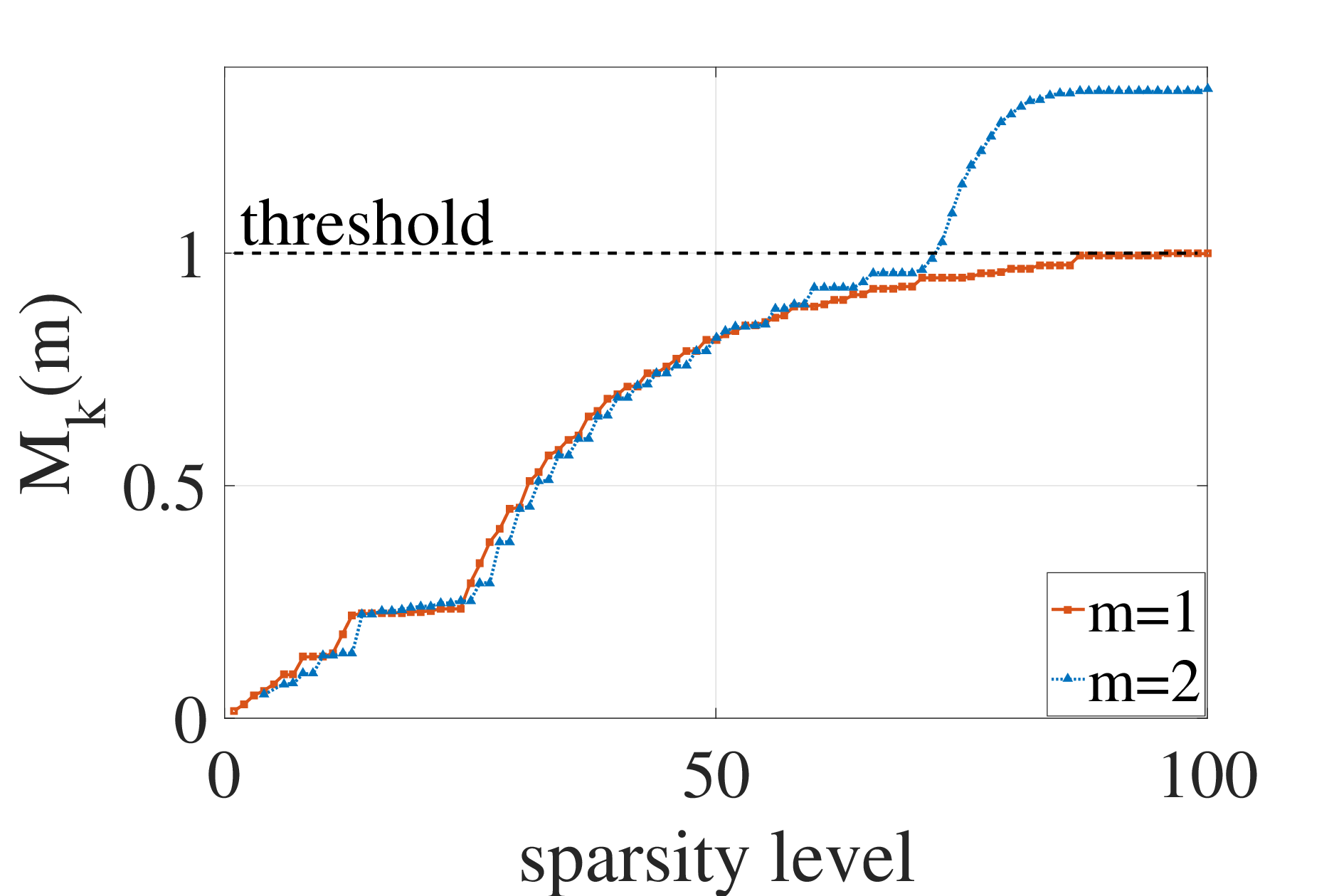}
        \caption{}
        \label{subfig:farey_mixture_erc}
    \end{subfigure}
    \hfill
    \begin{subfigure}[b]{0.46\textwidth}
    \centering
        \includegraphics[width = \textwidth,height = 6cm]{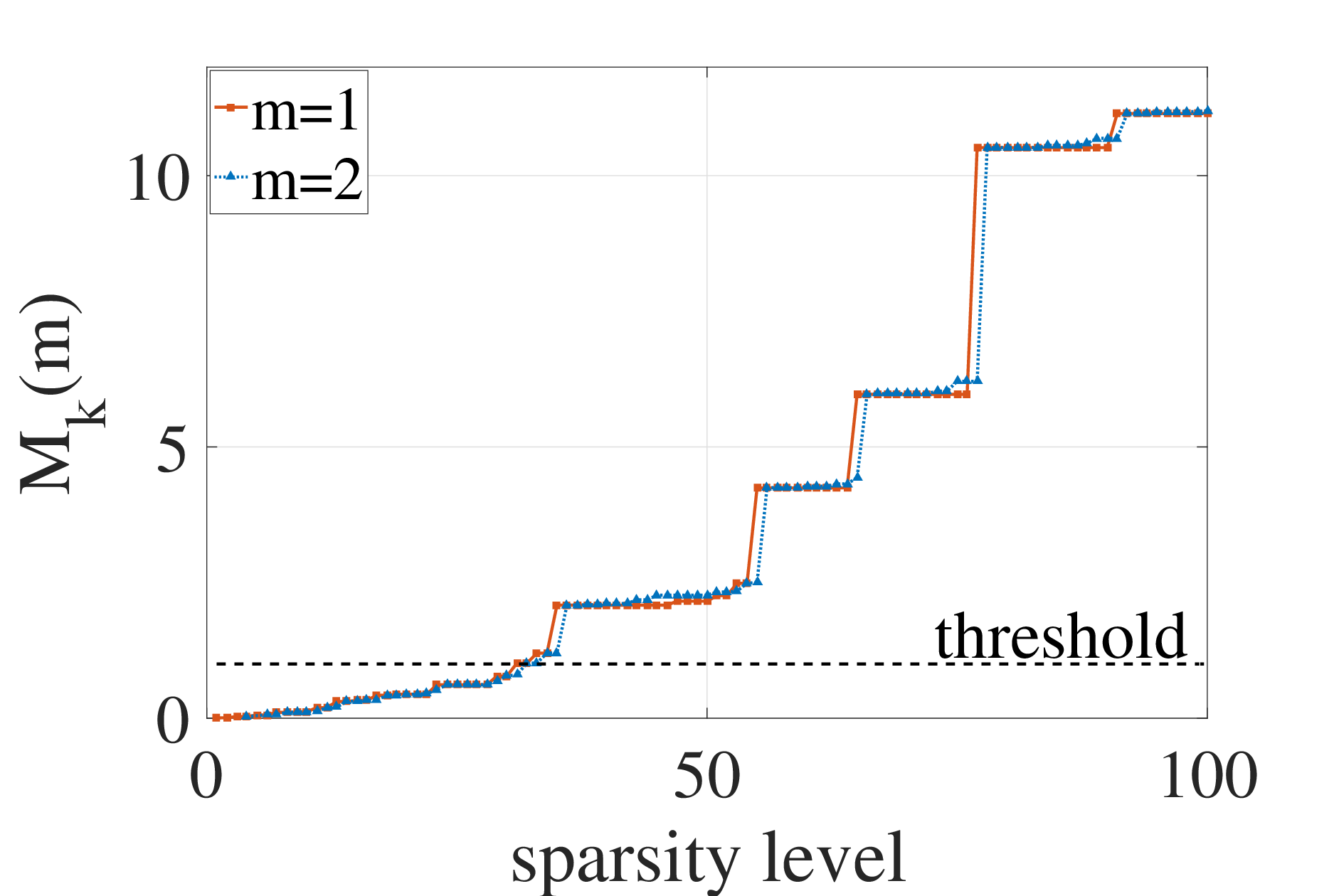}
    \caption{}
    \label{subfig:rpt_mixture_erc}
    \end{subfigure}
    \caption{Evaluation of the condition in (\ref{eq:ERC_mixture_general}) restricted to set $\mathbb{Q}_k\left(m\right)$ for (a) the Farey dictionary and (b) the RPT dictionary when $L = 1915$.}
    \label{fig:erc_mixture}
\end{figure}
\begin{figure}
    \centering
    \begin{subfigure}[b]{0.46\textwidth}
    \centering
        \includegraphics[width = \textwidth,height = 6cm]{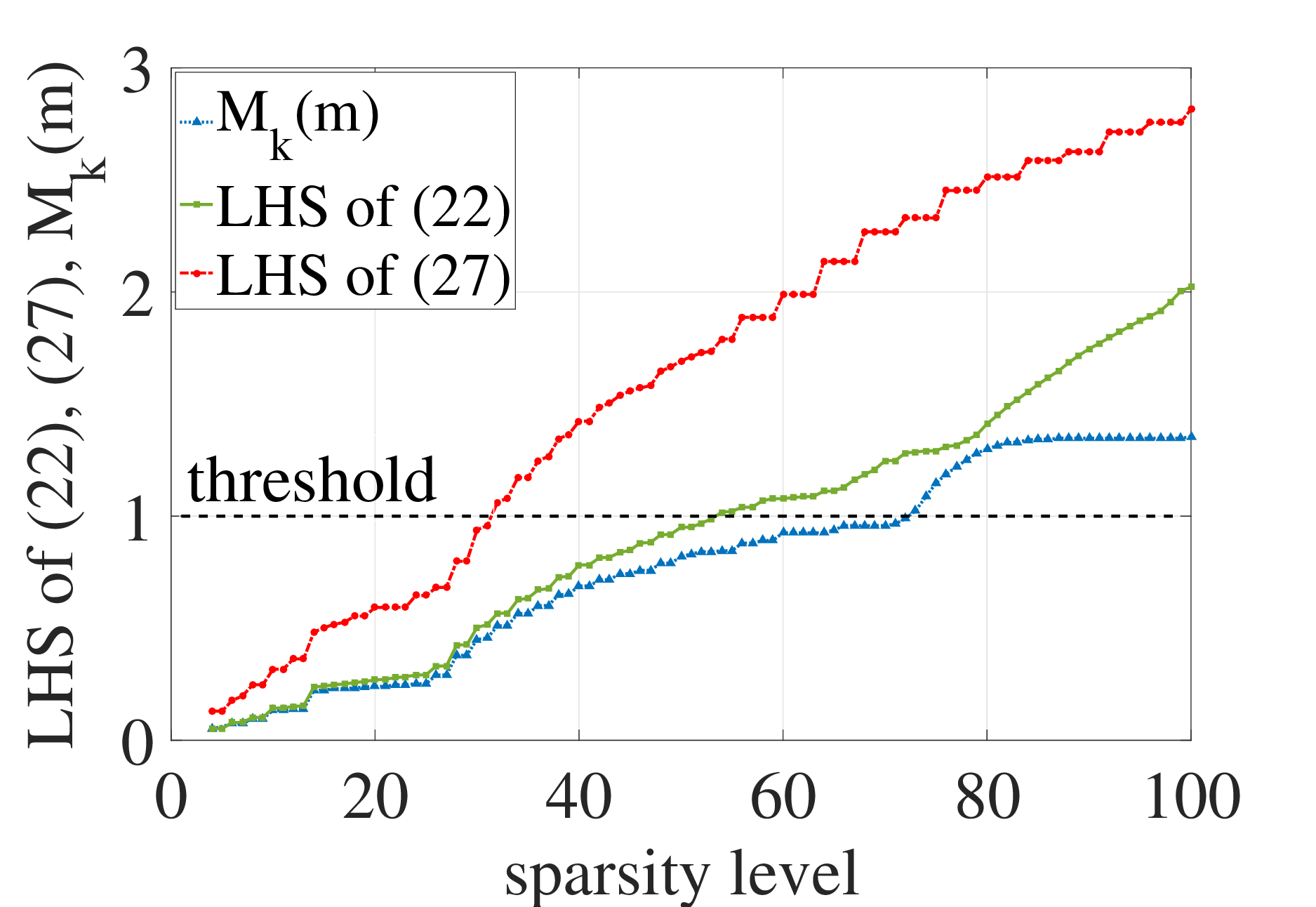}
        \caption{}
        \label{subfig:farey_mixture_bounds}
    \end{subfigure}
    \hfill
    \begin{subfigure}[b]{0.46\textwidth}
    \centering
        \includegraphics[width = \textwidth,height = 6cm]{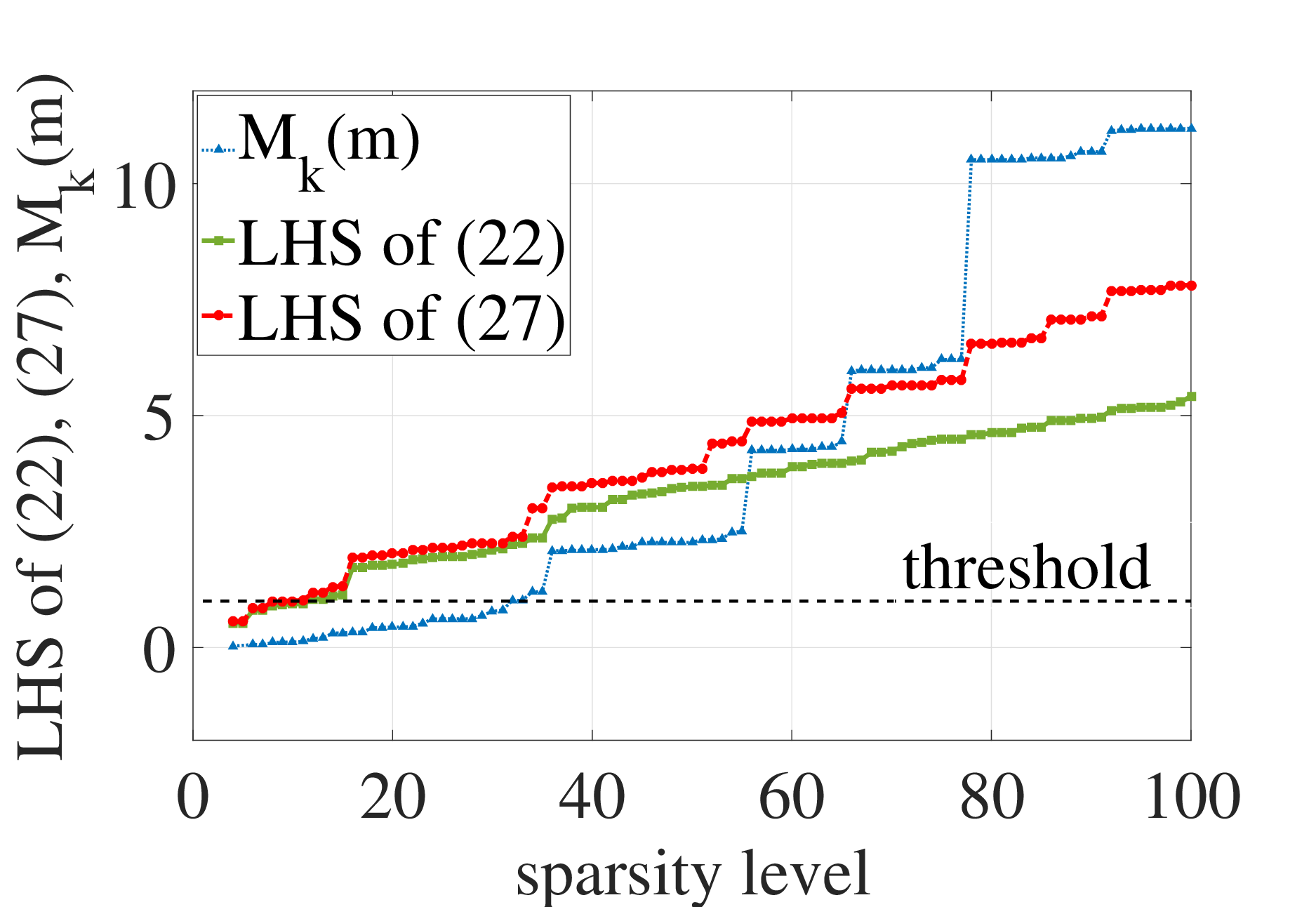}
    \caption{}
    \label{subfig:rpt_mixture_bounds}
    \end{subfigure}
    
   \begin{subfigure}[b]{0.46\textwidth}
   \centering
    \includegraphics[width = \textwidth,height = 6cm]{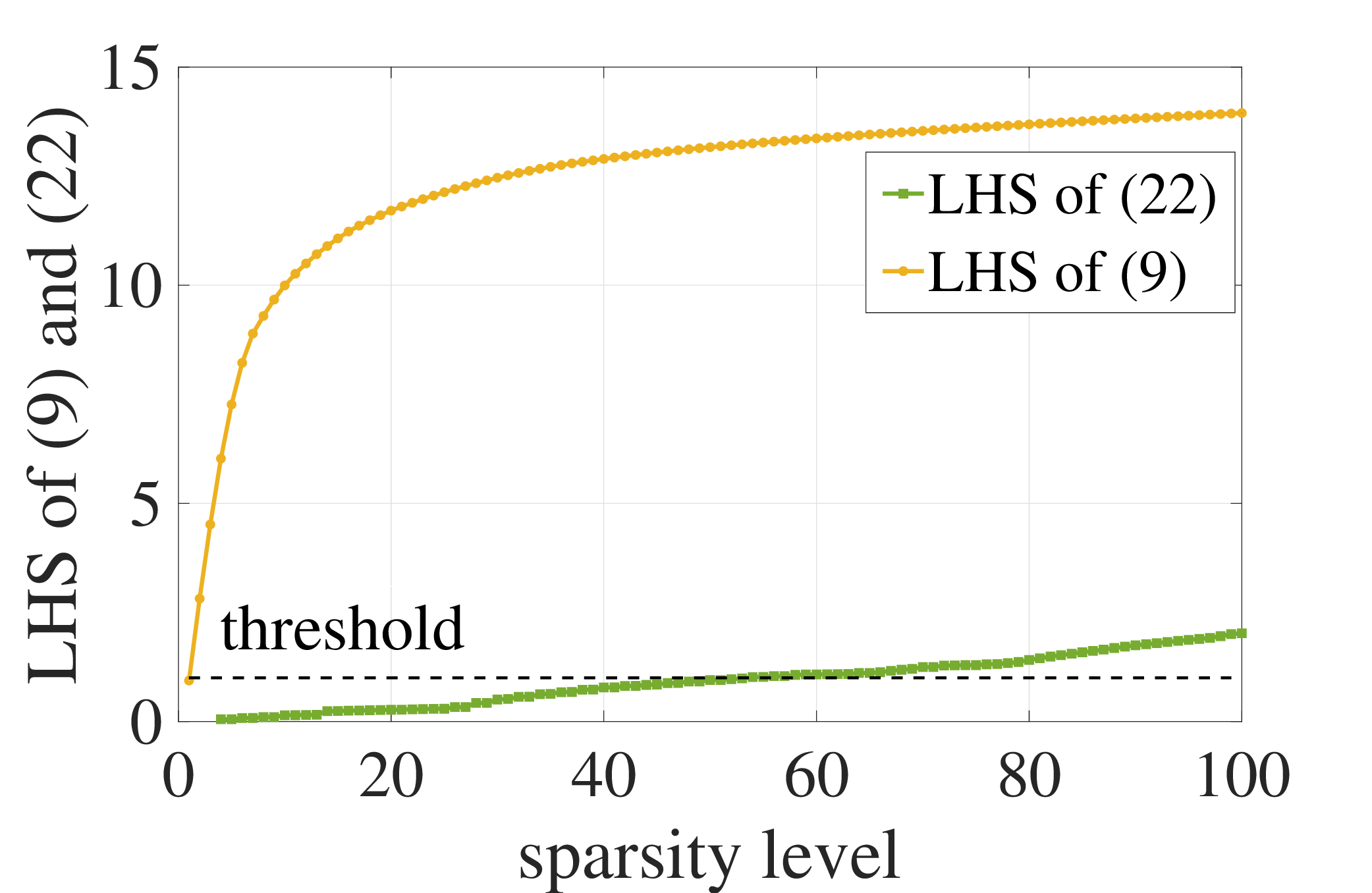}
    \caption{}
    \label{subfig:farey_mixture_bounds_old}
   \end{subfigure}
   \hfill
    \begin{subfigure}[b]{0.46\textwidth}
    \centering
    \includegraphics[width = \textwidth,height = 6cm]{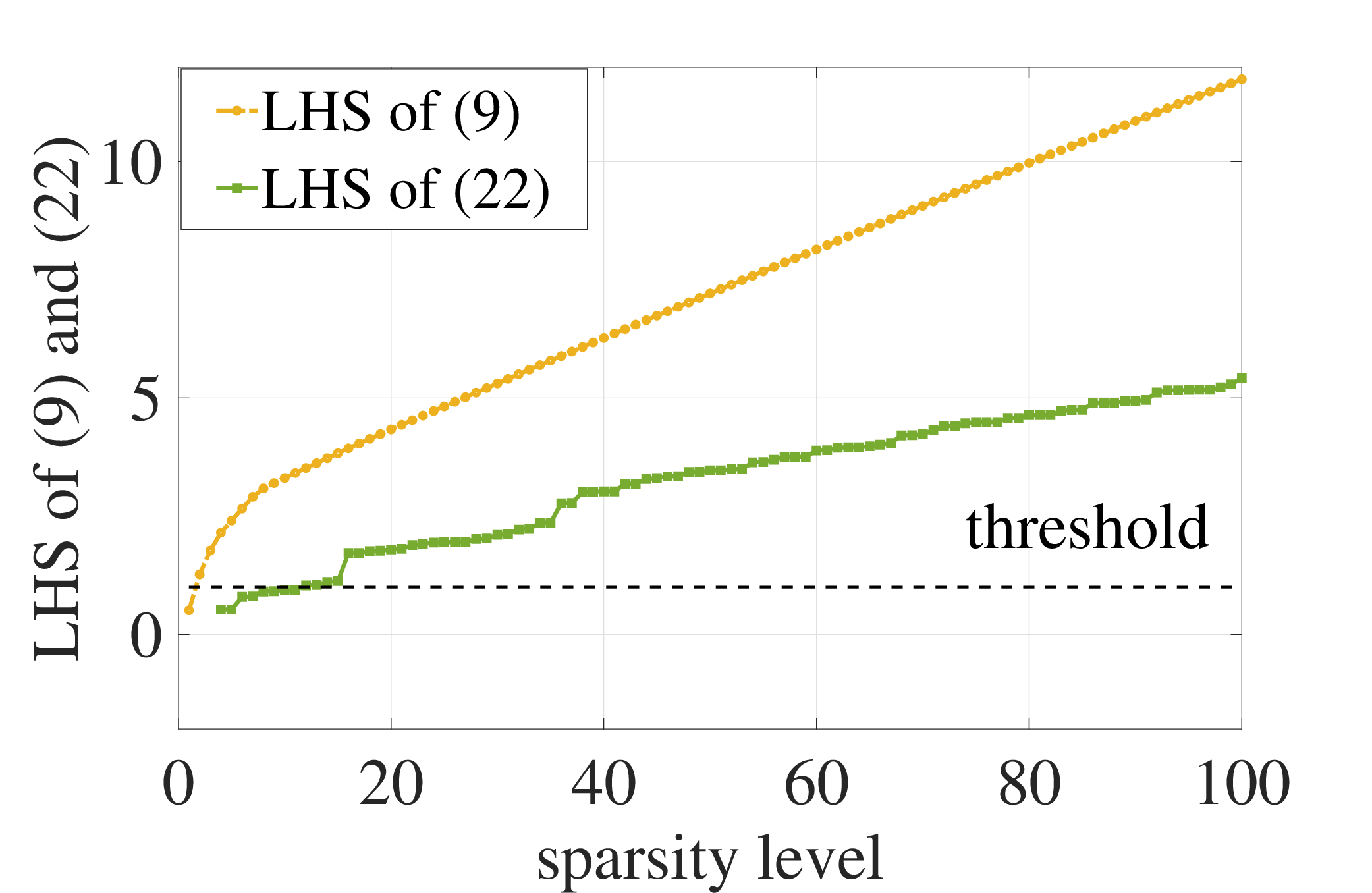}
    \caption{}
    \label{subfig:rpt_mixture_bounds_old}
    \end{subfigure}
    \caption{Comparison between $M_k(m)$ with $m=2$ and the conditions in \eqref{eq:guarantee_NPD_strong} and \eqref{eq:single_2_mixture_Qkm} for (a) the Farey dictionary and (b) the RPT dictionary. In addition (c) and (d)  show the improvement of \eqref{eq:guarantee_NPD_strong} over the bound in \eqref{eq:cumulative_coherence_guarantee} for the Farey and RPT dictionary, respectively. Here $P_{\max} = 100$ and $L = 1915$.}
    \label{fig:mixture_bounds}
\end{figure}
\begin{figure}
    \centering
    \begin{subfigure}[b]{0.46\textwidth}
    \centering        \includegraphics[width = \textwidth,height = 6cm]{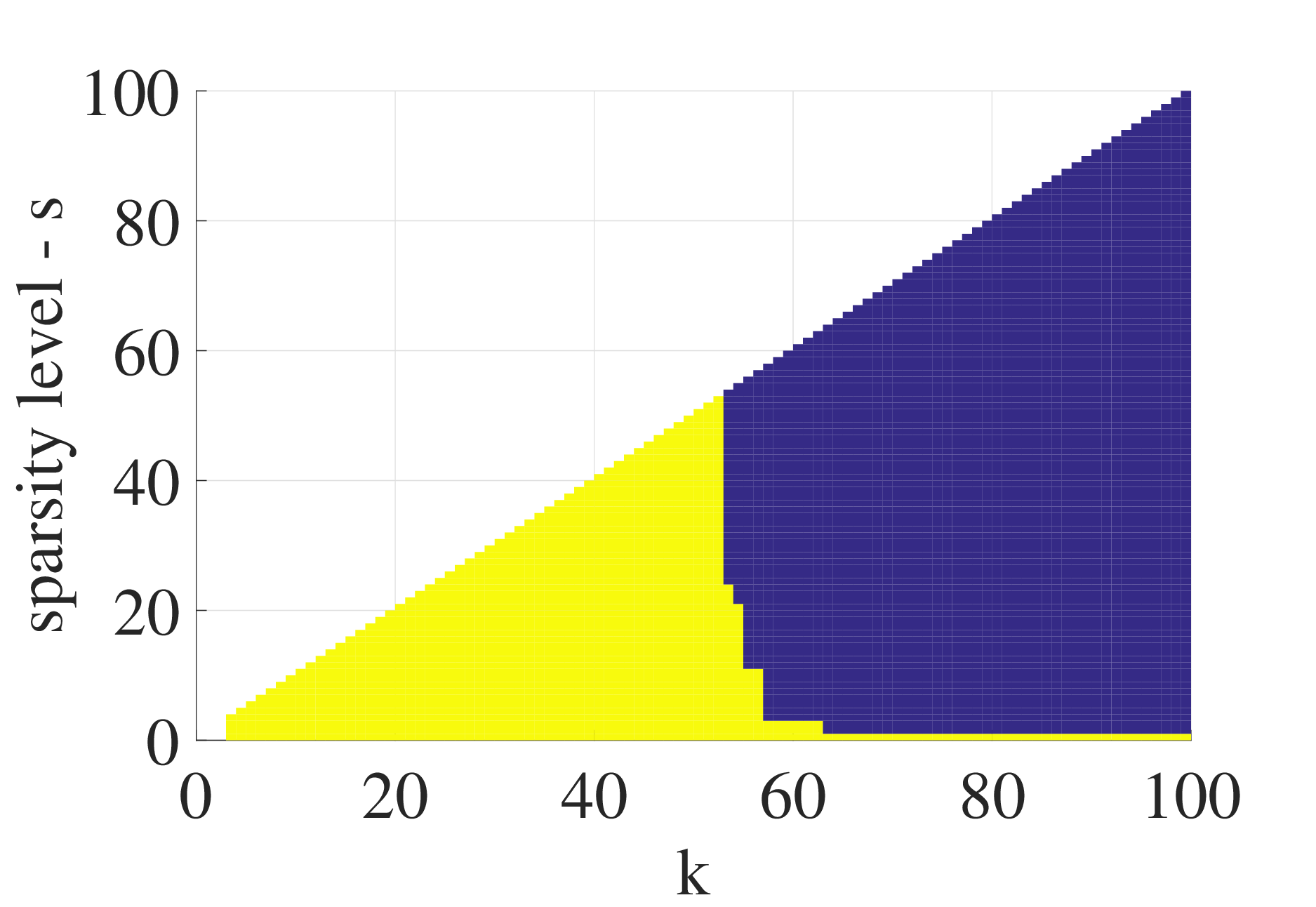}
        \caption{}
        \label{subfig:mixture_pt_dft}
    \end{subfigure}
    \hfill
    \begin{subfigure}[b]{0.46\textwidth}
    \centering
        \includegraphics[width = \textwidth,height = 6cm]{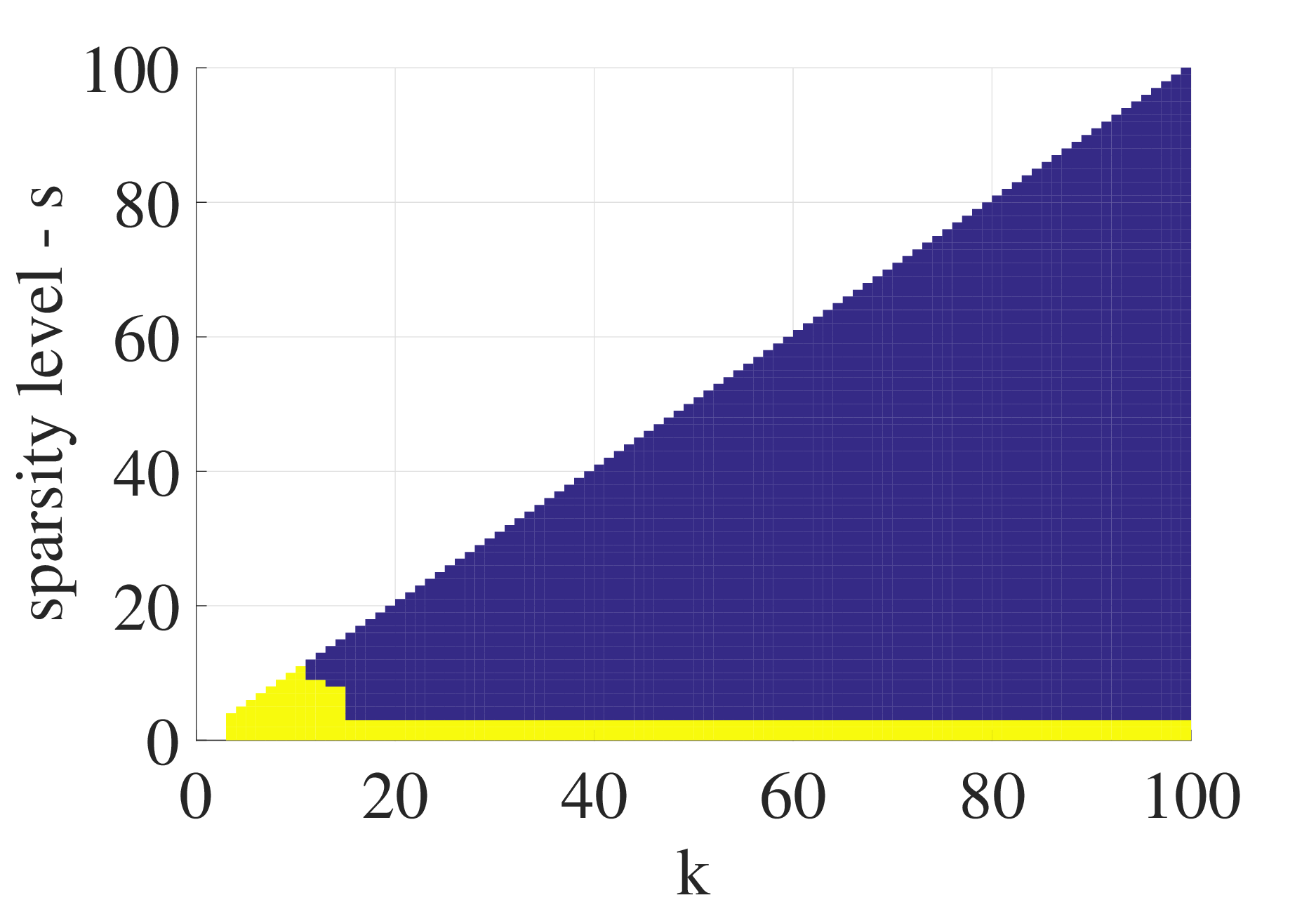}
    \caption{}
    \label{subfig:mixture_pt_rpt}
    \end{subfigure}
    \caption{Phase transition plots with respect to $k$ and sparsity level $s$ for periodic mixtures with $m = 2$, (a) for the Farey dictionary and (b) for the RPT dictionary when $P_{\max} = 100$ and $L = 1915$. The darker color indicates that the condition stated in \eqref{eq:guarantee_NPD_mixture_sparse} was not satisfied for that particular $k$ and $s$.
}
    \label{fig:mixture_pt}
\end{figure}
As shown in Fig. \ref{subfig:farey_mixture_erc}, $M_k\left(1\right) < 1$ for all $k \leq 100$ with the Farey dictionary, so the condition is met for all single period cases $p \in \mathbb{P}$. As a reminder, the sparsity level for a $p$-periodic signal under the NPDs is $\sum_{q|p} \phi\left(q\right) = p$. As $m$ increases, the achievability condition of the ERC condition given in Lemma \ref{lemma:mixture_general} ceases to hold at smaller values of $k$.
On the other hand, the bound is achievable over a smaller range of values of $k$ for the RPT dictionary as shown in Fig. \ref{subfig:rpt_mixture_erc}. 

Next, we examine the implications of Theorem \ref{thr:guarantee_upperbound_mixture_strong} and Corollary \ref{cor:single_2_mixture}, where we derived two alternative bounds on the recovery condition tailored to set $\mathbb{Q}_k\left(m\right)$ for $m = 2$. In Fig. \ref{fig:mixture_bounds}, 
we compare the new improved bounds in \eqref{eq:guarantee_NPD_strong} and \eqref{eq:single_2_mixture_Qkm} with $M_k\left(m\right)$ and the bound in \eqref{eq:cumulative_coherence_guarantee} for the Farey and RPT dictionary. We choose $P_{\max} = 100$ and $L = 1915$. As shown in Fig. \ref{subfig:farey_mixture_bounds} and Fig. \ref{subfig:rpt_mixture_bounds}, the conditions of Theorem \ref{thr:guarantee_upperbound_mixture_strong} and Corollary \ref{cor:single_2_mixture} are tight upper bounds on $M_k\left(m\right)$ when the left hand side (LHS) in \eqref{eq:guarantee_NPD_strong} and \eqref{eq:single_2_mixture_Qkm} are smaller than $1$. As shown in the proof in Appendix \ref{App:proof_strong_mixture} and Appendix \ref{App:proof_single_2_mixture}, these bounds are valid if $\nu_{k,m} < 1$ and $\max_{p_j\in\mathcal{T}}\nu_{p_j}+\sum_{p_j \in \mathcal{T}}\zeta_{p_j} - \min_{p_j\in\mathcal{T}}\zeta_{p_j}<1$ in \eqref{eq:guarantee_NPD_strong} and \eqref{eq:single_2_mixture_Qkm}, respectively. Therefore, for some values of $k$, the bounds are invalid as can be seen in Fig. \ref{subfig:rpt_mixture_bounds}. On the other hand, \eqref{eq:guarantee_NPD_strong} provides significant improvement over the stringent condition in \eqref{eq:cumulative_coherence_guarantee}, which is only achievable for sparsity level $k = 1$ for both the RPT and Farey dictionary, as shown in Fig. \ref{subfig:farey_mixture_bounds_old} and Fig. \ref{subfig:rpt_mixture_bounds_old}.

Furthermore, Table \ref{tab:runtime} shows the required run times to compute and verify the ERC-based bound, i.e., $M_k\left(m\right)$, the bound in Theorem \ref{thr:guarantee_upperbound_mixture_strong} based on NPA and NPI, and the computationally efficient bound for both the Farey and RPT dictionary and all the elements in the set $\mathbb{Q}_k\left(m\right)$, when $m = 2$, and $1 \leq k \leq 100$.
\begin{table}
    \centering
    \caption{Run time to establish the exact recovery conditions in Theorem 1, Corollary 3, and $M_k\left(m\right)$ for $m = 2$ and $1 \leq k \leq 100$ for the Farey and RPT dictionaries.}
    \label{tab:runtime}
    \begin{tabular}{ccc}
        Method & Run time (seconds) & Run time (seconds)\\
        & Farey dictionary &RPT dictionary\\ \hline
        $M_k\left(m\right)$ & 521.84  & 234.84\\ \
        LHS of \eqref{eq:guarantee_NPD_strong} & 216.77 & 121.54\\
        LHS of \eqref{eq:single_2_mixture_Qkm} & 31.18 & 24.64
    \end{tabular}
\end{table}
Next, we investigate the implications of Theorem \ref{thr:guarantee_mixture_sparse_based}, which gives a refined recovery bound when the induced sparsity level is assumed to be $s$ for some $s \leq k$. Fig. \ref{fig:mixture_pt} shows phase transition plots with respect to different values of $k$  and sparsity level $s \leq k$. Fig. \ref{subfig:mixture_pt_dft} shows that for $k<54$, the condition \eqref{eq:guarantee_NPD_mixture_sparse} is met for all $s\leq k$ with the Farey dictionary, which is in agreement with the results of Theorem \ref{thr:guarantee_upperbound_mixture_strong} depicted in Fig. \ref{subfig:farey_mixture_bounds}. Furthermore, it shows that for $k\geq54$, BP and OMP can recover the exact sparse vector for some $s\leq k$, and always, when $s = 1$ in agreement with the results in Fig. \ref{subfig:farey_mixture_erc}. Similarly, Fig. \ref{subfig:mixture_pt_rpt} shows the phase transition plot for the RPT dictionary. It shows that for $k<12$, the condition in \eqref{eq:guarantee_NPD_mixture_sparse} is met for all $s\leq k$, and for all $k\leq 100$, BP and OMP can recover the exact sparse vector of a periodic mixture with $m = 2$, and $s \leq 3$.
Lastly, we validate these results numerically, by recovering the sparse representations of periodic mixtures using the BP and OMP techniques. To this end, we generate $100$ periodic mixtures for each element in $\mathbb{Q}_k\left(2\right)$, and use both methods to recover the sparse vectors in \eqref{eq:NPD_model}. We choose the RPT dictionary with $P_{\max} = 20$ and $L = 100$. To evaluate performance, we compute the root mean square error (RMSE) between the recovered $\hat{\bx}$ and the original vector $\bx$ that is the solution to the restricted mean square error problem.
Fig. \ref{subfig:guarantee_bound_rpt} and \ref{subfig:guarantee_bound_dft} show the phase transitions with respect to $k$ and $s$ for the RPT dictionary and the Farey dictionary with $P_{\max} = 20$ and $L = 100$, respectively. As shown, the sufficient condition to recover the exact sparse vector holds for $k\leq 5$ in the case of the RPT dictionary and for $k \leq 11$ for the Farey dictionary. Fig. \ref{subfig:numerical_rpt} and Fig. \ref{subfig:numerical_dft} show the average RMSE versus $k$ and verify perfect recovery when $k\leq 5$ for the RPT dictionary and $k \leq 11$ for the Farey dictionary for both BP and OMP.
\begin{figure}
    \centering
    \begin{subfigure}[b]{0.46\textwidth}
    \centering
        \includegraphics[width = \textwidth,height = 6cm]{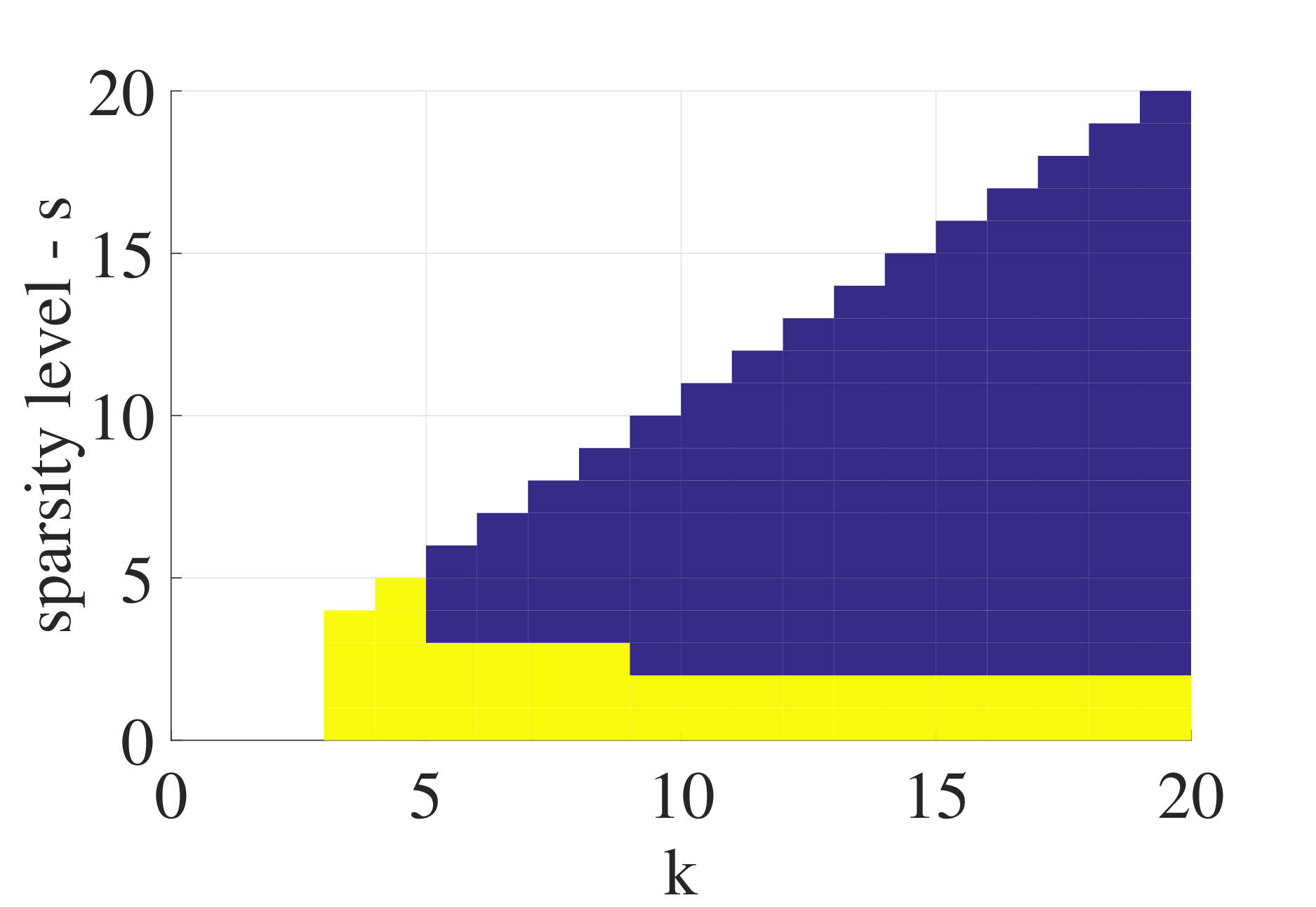}
        \caption{}
        \label{subfig:guarantee_bound_rpt}
    \end{subfigure}
    \hfill
    \begin{subfigure}[b]{0.46\textwidth}
    \centering
        \includegraphics[width = \textwidth,height = 6cm]{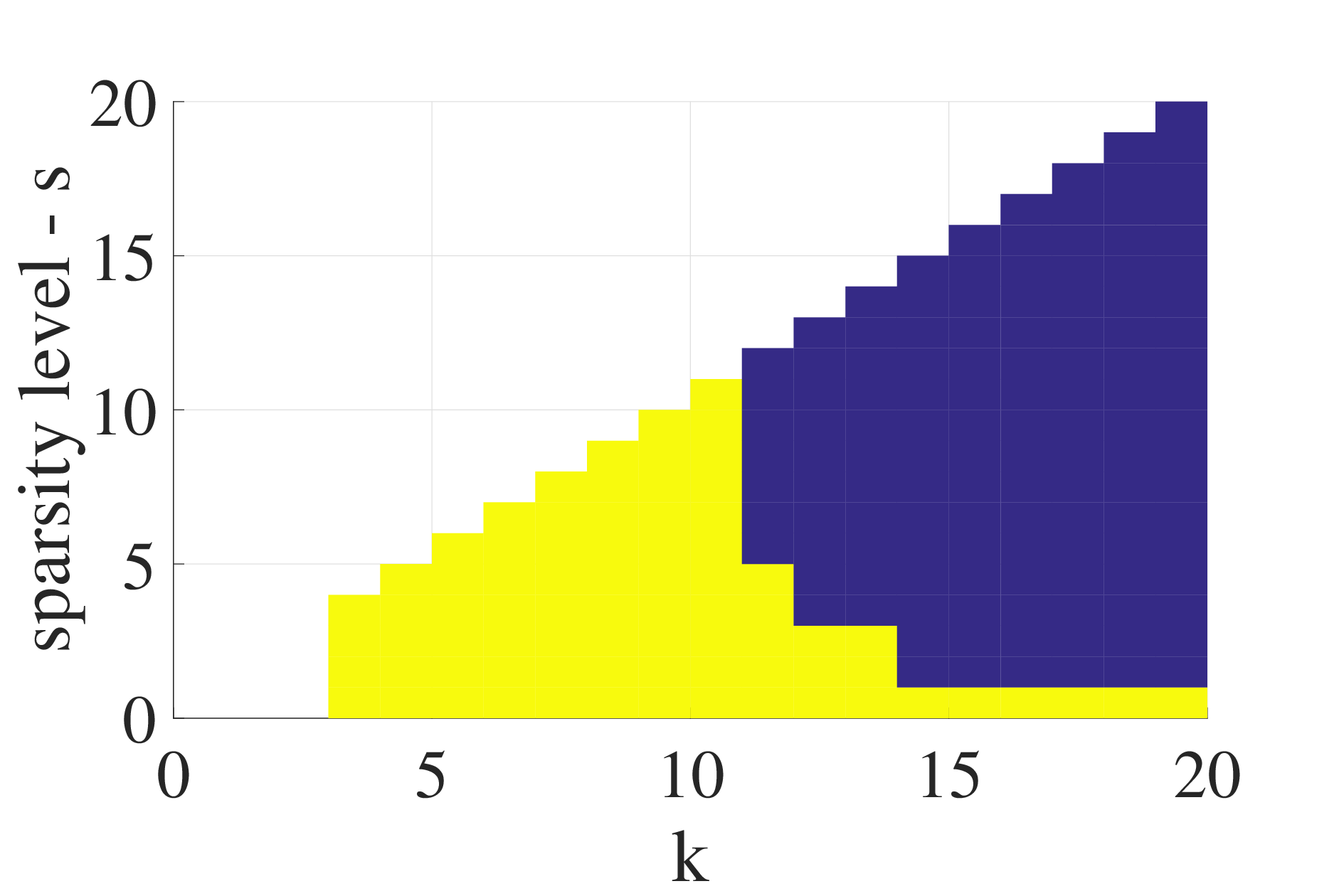}
    \caption{}
    \label{subfig:guarantee_bound_dft}
    \end{subfigure}
     \begin{subfigure}[b]{0.46\textwidth}
    \centering
        \includegraphics[width = \textwidth,height = 6cm]{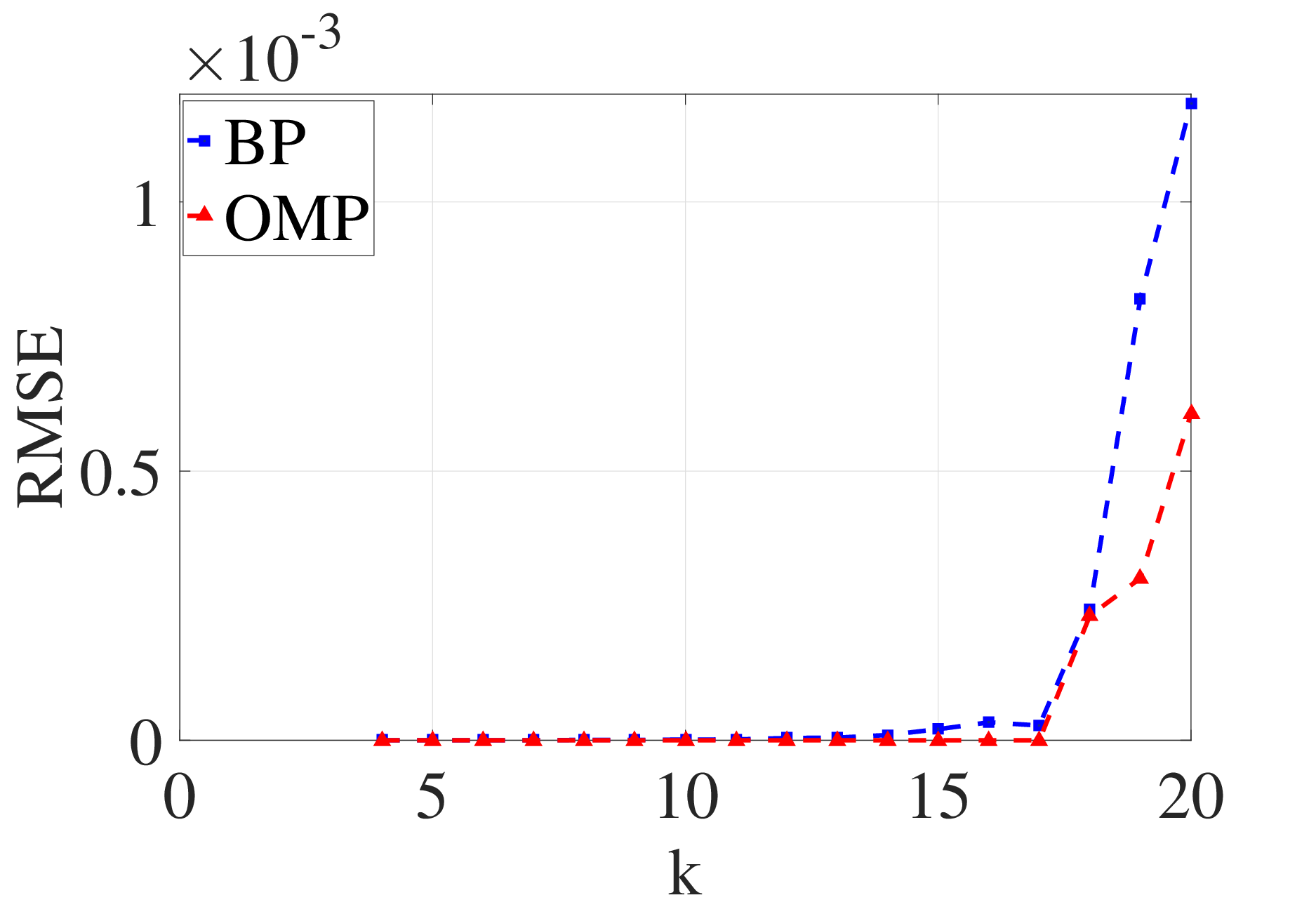}
    \caption{}
    \label{subfig:numerical_rpt}
    \end{subfigure}
    \hfill
    \begin{subfigure}[b]{0.46\textwidth}
    \centering
        \includegraphics[width = \textwidth,height = 6cm]{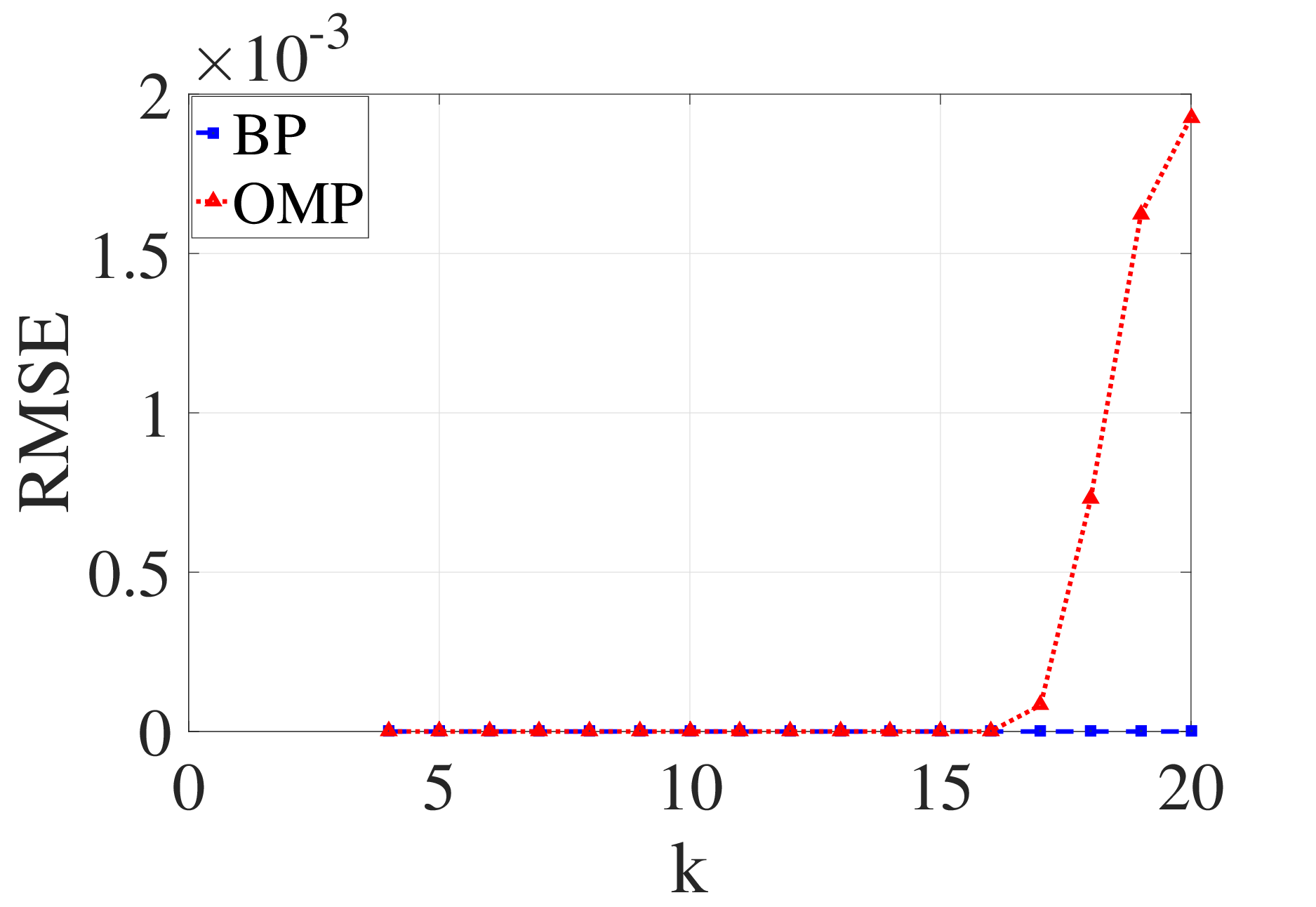}
    \caption{}
    \label{subfig:numerical_dft}
    \end{subfigure}
    \caption{Phase transition plots with respect to $k$ and $s$ for (a) the RPT dictionary (left above) and (b) the Farey dictionary (right above), with $P_{\max} = 20$ and $L = 100$. Average RMSE for BP and OMP algorithms (c) for the RPT dictionary (left below) and (d) the Farey dictionary (right below).}
    \label{fig:numerical}
\end{figure}

\subsection{Exact support recovery condition in the presence of noise} \label{subsec:result_noisy}
\subsubsection{Bounded noise} To validate the results of Theorem \ref{thr:noise_bounded_general_zeta_nu} and \ref{thr:noise_bounded_general_s2m},
we choose the RPT dictionary with $P_{\max} = 20$, $L = 100$ and set $\epsilon = 0.5$. We consider a periodic signal with period $4$. We verified the conditions \eqref{eq:guarantee_NPD_strong} and \eqref{eq:single_2_mixture} and ensured they hold. Then, we randomly generated $2000$ sparse signals that are supported only on $S_{4}$, and ensured that the absolute values of the nonzero coefficients $|x_i|$ are greater than some threshold $\gamma$ by adding (or subtracting) $\gamma$ to (from) $x_i$. Then, based on the model in \eqref{eq:NPD_model_noisy}, we constructed $2000$ periodic and noisy signals, and used the OMP algorithm to recover the underlying support. We increase $\gamma$ from 0 until it meets the condition in \eqref{eq:noise_bounded_s2m}. At each step, we calculate the RMSE and success rate. Recovery is deemed successful only if the recovered support set is identical to $S_4$. Fig. \ref{subfig:noise_bounded_rpt4_zeta_nu_success} and \ref{subfig:noise_bounded_rpt4_zeta_nu_rmse} show the success rate and RMSE, respectively. As shown, the RMSE is $0$ and success rate is $1$ for $\gamma$ greater than the right hand side (RHS) of \eqref{eq:noise_bounded_s2m} that is $1.21$. Note that, given the setting above, the RHS of \eqref{eq:omp_npd_noise_bnd_zeta_nu} is $6.72$, suggesting a more restrictive condition. We remark that for a periodic signal with period $4$ that induces a sparsity level  $k = 4$, the condition in \cite[Theorem 1]{cai2011orthogonal} requires $\mu<\frac{1}{2k-1}=0.1429$. However, for the described dictionary $\mu = 0.5285$. Therefore, the existing conditions are of little use in the context of NPDs.

\begin{figure}
    \centering
    \begin{subfigure}[b]{0.46\textwidth}
    \centering
        \includegraphics[width = \textwidth,height = 6cm]{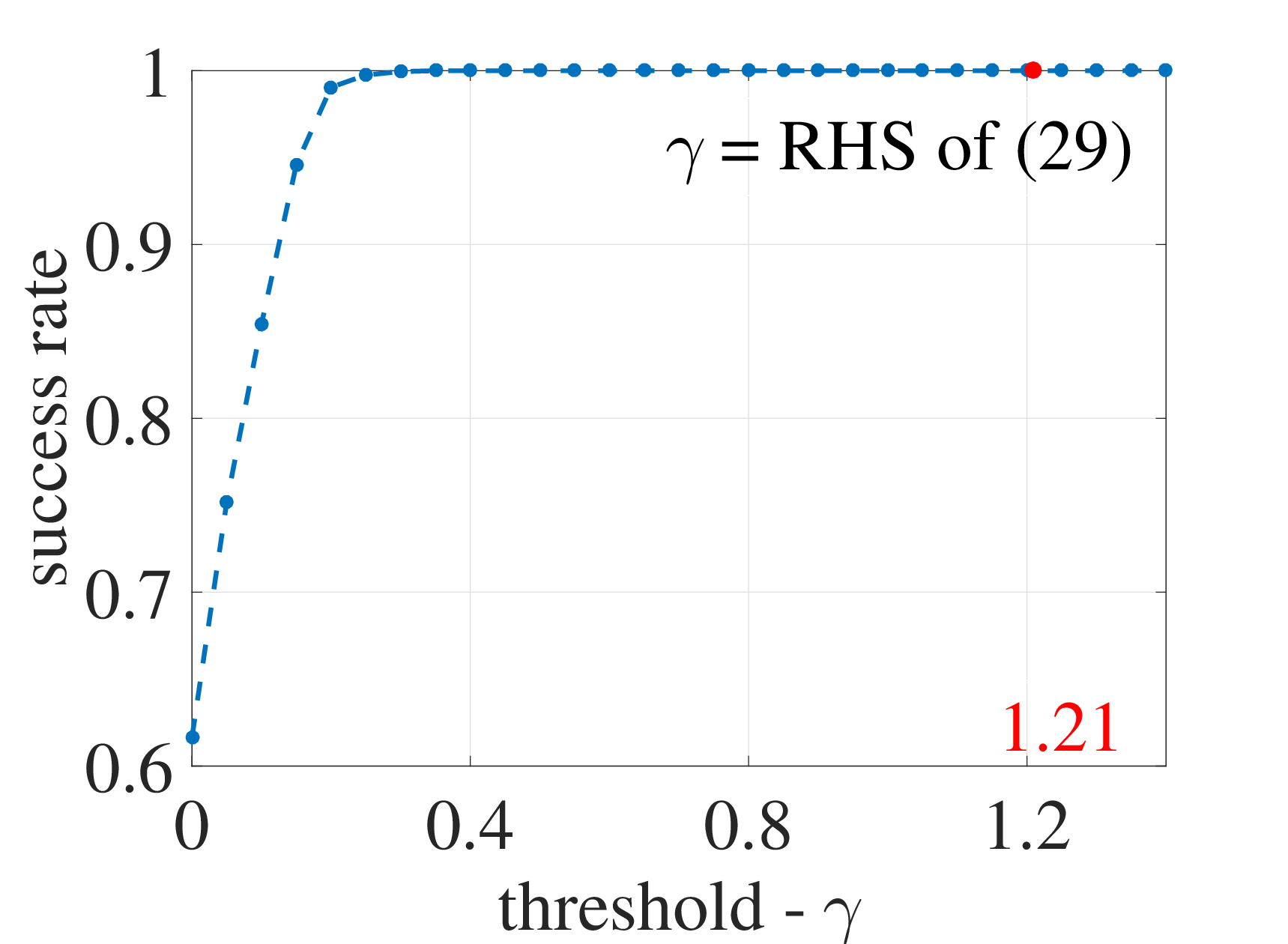}
        \caption{}
        \label{subfig:noise_bounded_rpt4_zeta_nu_success}
    \end{subfigure}
    \hfill
    \begin{subfigure}[b]{0.46\textwidth}
    \centering
        \includegraphics[width = \textwidth,height = 6cm]{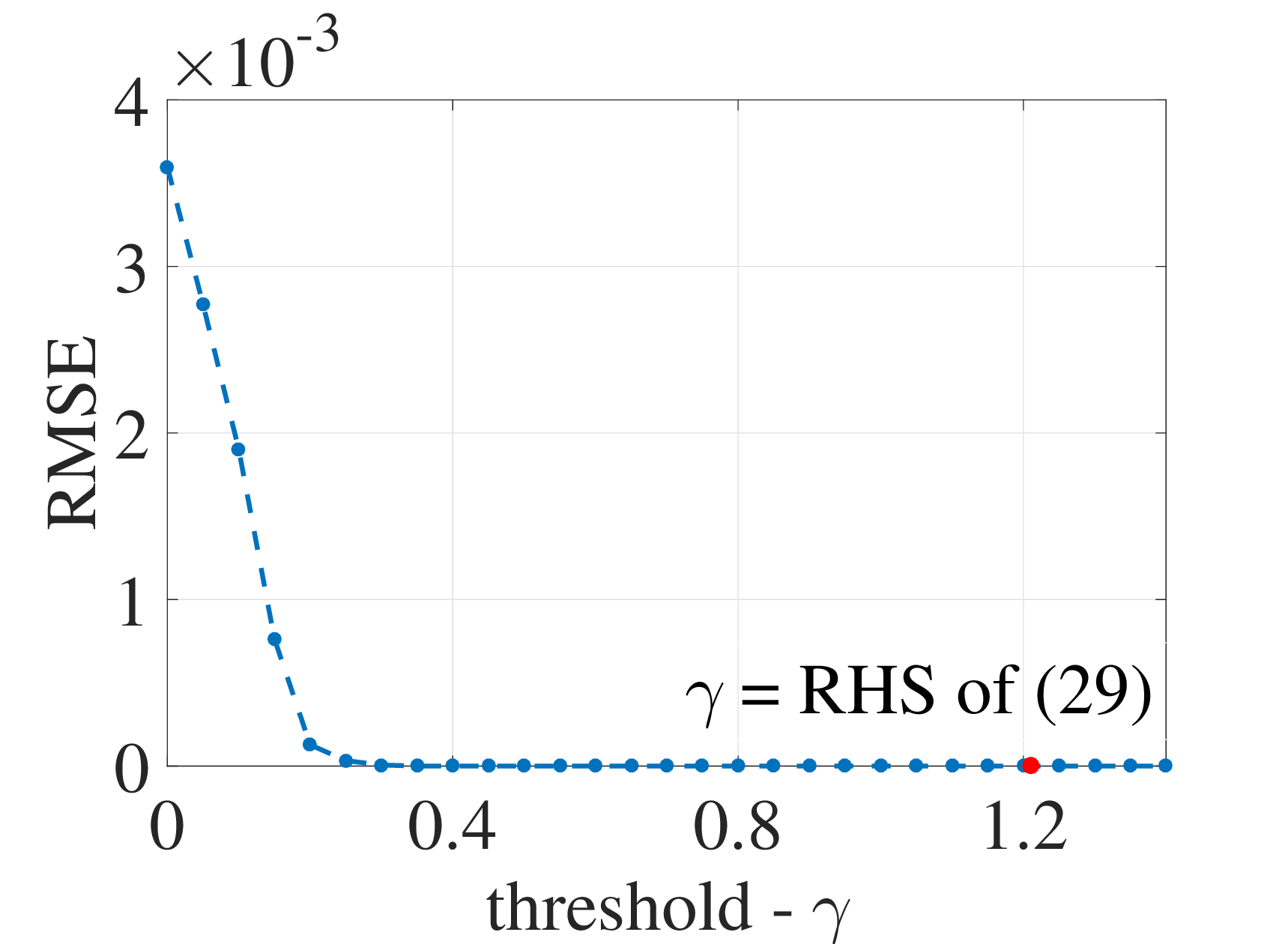}
    \caption{}
    \label{subfig:noise_bounded_rpt4_zeta_nu_rmse}
    \end{subfigure}
    \caption{(a) The success rate for the OMP algorithm in the presence of the $\ell_2$-norm bounded noise. (b) The RMSE. Here, $P_{\max} = 20$ and $L = 100$ and the RPT dictionary is used.}
    \label{fig:noise_bounded_rpt4}
\end{figure}

\subsubsection{Gaussian noise}
\textcolor{black}{In this section, we validate the results of Theorem \ref{thr:noise_gaussian_s2m}. For this purpose, we choose the RPT dictionary, with $P_{\max} = 20$ and $L = 100$. We randomly generate periodic signals with period $4$ by first constructing the sparse vector $\bx$ supported on $S_4$ such that the nonzero coefficients obey
\begin{equation}
\begin{aligned}
    \label{eq:gaussian_coeff}
    &|x_i| \geq\\ 
    &\alpha\frac{2\sigma \sqrt{L + 2\sqrt{L\log L}}}{\left(1-2\sum_{p_j\in \mathcal{T}} \zeta_{p_j} - \hat{\nu}_{p} + \check{\zeta}_p\right)\hspace{-3pt}\left(1-\sum_{p_j\in \mathcal{T}} \zeta_{p_j}-\hat{\nu}_p + \check{\zeta}_p\right)}\:
    \end{aligned}
\end{equation}
for some $0 \leq \alpha \leq 1$. For $\alpha = 1$, the condition agrees with that 
in Theorem \ref{thr:noise_gaussian_s2m}. We gradually increase $\alpha$ by a step-size of $0.01$ and generate $2000$ periodic signals at each step based on the model in \eqref{eq:NPD_model_noisy}. Then, we use the OMP algorithm to recover the support set of $\bx$. Recovery is successful if the support set is exactly recovered, 
otherwise it counts as an error. In addition, we recover the sparse vector $\bx$ and calculate the RMSE, shown in Fig. \ref{fig:noise_gaussian_s2m}. When the nonzero coefficients obey the condition in \eqref{eq:gaussian_coeff_s2m}, the OMP algorithm was able to recover the exact support with probability greater than $1 - 1/L$. 
\begin{figure}
    \centering
    \begin{subfigure}[b]{0.46\textwidth}
    \centering
        \includegraphics[width = \textwidth,height = 6cm]{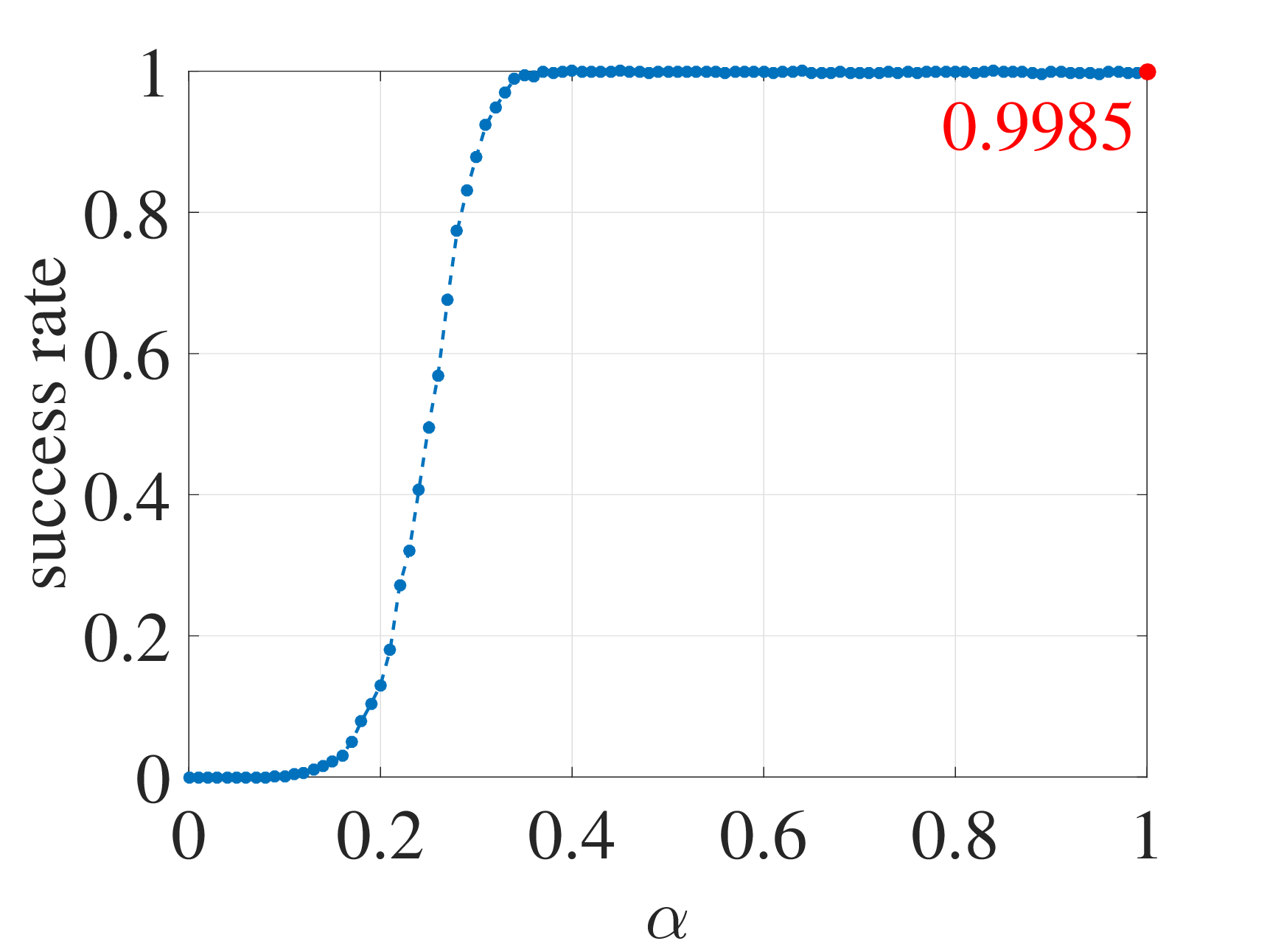}
        \caption{}
        \label{subfig:noise_gaussian_successrate}
    \end{subfigure}
    \hfill
     \begin{subfigure}[b]{0.46\textwidth}
    \centering
        \includegraphics[width = \textwidth,height = 6cm]{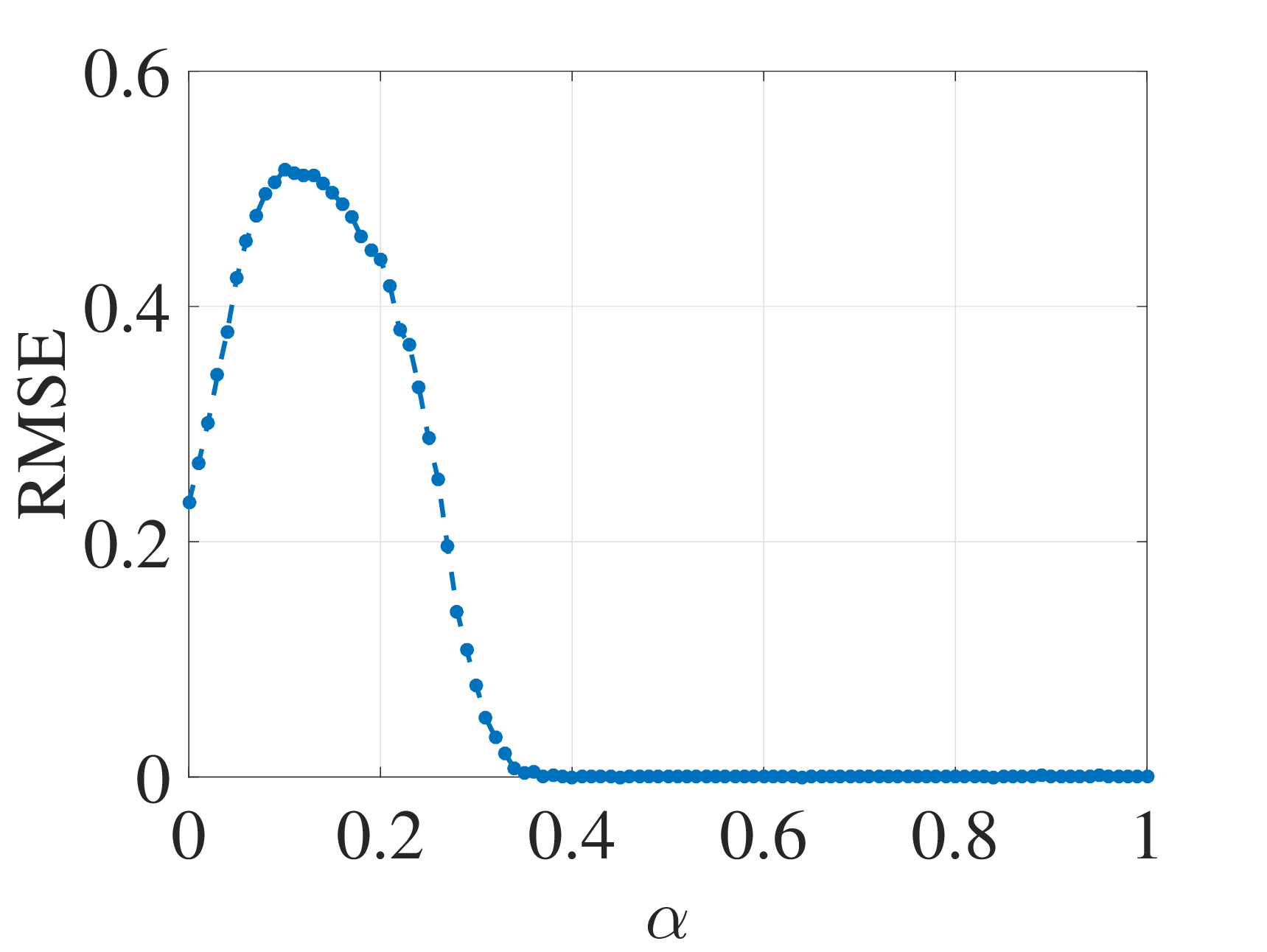}
        \caption{}
        \label{subfig:noise_gaussian_rmse}
    \end{subfigure}
    \caption{(a) Success rate for the OMP algorithm in the presence of Gaussian noise. (b) The RMSE. Here, $P_{\max} = 20$ and $L = 100$ and the RPT dictionary is used.}
    \label{fig:noise_gaussian_s2m}
    
\end{figure}}


\section{Discussion and Conclusion}
\label{sec:discussion}
Existing support recovery bounds such as \eqref{eq:cumulative_coherence_guarantee} are fairly tight for random and unstructured dictionaries. However, they fall short of predicting meaningful achievability regions with NPDs since they ignore their special structure. In this paper, we obtained improved recovery conditions by accounting for special properties intrinsic to NPDs, such as the Euler and the LCM properties. Leveraging this structure, we derived sufficient conditions for exact recovery of periodic mixtures admitting sparse representations in an NPD, and in turn of the underlying hidden periods. Our findings reveal that a family of NPDs that span the Ramanujan subspaces -- instances of nested periodic subspaces --  such as the Farey and the RPT dictionaries, satisfy the newly derived conditions for a wide range of the sparsity level $k$. Numerical results suggest that the Farey dictionary could be a better choice than the RPT dictionary with regard to meeting the sufficient conditions for exact support recovery, albeit the latter 
has lower computational complexity.

While the derived conditions improve significantly upon generic sparse recovery bounds, the numerical results have shown that a gap remains between the theoretic achievability bounds and the actual performance of recovery algorithms. Closing this gap is an avenue for future research. 
Moreover, a study of the block sparse recovery problem with NPDs for block sparse periodic signals is a subject of future work.

\bibliographystyle{IEEEbib}
\bibliography{bibfile}

\newpage
\appendices
\section{} 
\label{App:General_ERC_NPD_mixture}
First, we restate the classical ERC condition from \cite[Theorem A]{tropp2004greed} that is essential to prove the results in this paper.
\begin{lemma}
\label{lem:original_ERC}
(\cite[Theorem  A]
{tropp2004greed}). Given a dictionary $\mathbf{\Phi}$ of size $L \times N$, and $N>>L$, a sufficient condition for OMP and BP to recover the sparsest representation of the input signal is that
\begin{equation}
    \max_{\mathbf{\psi}} \|\left(\mathbf{\Phi}^{H}_{\opt}\mathbf{\Phi}_{\opt}\right)^{-1}\mathbf{\Phi}_{\opt}^{H} \mathbf{\psi}\|_1 < 1 \:,
\end{equation}
where $\mathbf{\Phi}_{\opt}$ is the $L \times k$ matrix containing only the columns from $\mathbf{\Phi}$ that form the optimal representation of the signal, and $\mathbf{\psi}$ ranges over the atoms in $\mathbf{\Phi}$ that are not part of $\mathbf{\Phi}_{\opt}$.
\end{lemma}
Next, we state the following Lemma that connects the ERC condition and the period estimation problem and is used in the proof of Theorem \ref{thr:guarantee_upperbound_mixture_strong}, \ref{thr:single_2_mixture} and \ref{thr:guarantee_mixture_sparse_based}.

\begin{lemma} \label{lemma:mixture_general}
Suppose the matrix $\bK$ is an NPD of size $L\times N$, where $N = \sum_{p=1}^{P_\mathrm{max}} \phi\left(p\right)$. Let $\mathcal{T} \in 2^{\mathbb{P}}$ be a non-empty set that contains the hidden periods of a periodic mixture $\by$ and $S_{\mathcal{T}}$ the set that contains the support set of $\bx$ in \eqref{eq:NPD_model} as in Definition \ref{def:union_support_set}. If the following condition holds
\begin{equation}
    \label{eq:ERC_mixture_general}
    \begin{aligned}
        \|\left(\bK_{S_{\mathcal{T}}}^H\bK_{S_{\mathcal{T}}}\right)^{-1}\bK_{S_{\mathcal{T}}}^H\bK_{S_{\mathcal{T}}^c}\|_{1,1} <1,
    \end{aligned}
\end{equation}
then one can identify all the hidden periods in $\by$ by recovering the support set of $\bx$ using BP or OMP.
\end{lemma}
\begin{proof}
Let $\by$ be a periodic mixture with hidden periods $\mathcal{T} = \{p_1,p_2,\ldots,p_m\}$ from the set $\mathbb{P}$.

Let $\mathbb{K}$ be an extended version of $\bK$ with infinite length, and $\mathbb{K}_{S_\mathcal{T}}$ spans the periodic subspaces $\mathcal{V}_{q|p_i}$ for all $p_i \in \mathcal{T}$. This implies that if $\by$ was also infinitely long, then $\tilde{S}\subseteq S_{\mathcal{T}}$ would be the support set of the sparsest vector $\bx$. Here, $S_{\mathcal{T}}$ is defined in Definition \ref{def:union_support_set}. As $\bK$ is an overcomplete dictionary of size $L \times N$, the ERC in Lemma \ref{lem:original_ERC} states that the sparse vector $\bx$ can be exactly recovered using BP and OMP if
\begin{equation}
\label{eq:ERC_tropp}
\begin{aligned}
M = \max_{j\in \tilde{S}^c} |\left(\bK_{\tilde{S}}^H\bK_{\tilde{S}}\right)^{-1}\bK_{\tilde{S}}^H\bk_j|1 <1 \:,
\end{aligned}
\end{equation}
where $\bK_{\tilde{S}}$ is the submatrix of $\bK$ that consists of the columns corresponding to $\tilde{S}$, and $\bk_j$ is the $j$-th column of $\bK$. This condition allows for the estimation of the hidden periods using the LCM property.

For the case where $S_\mathcal{T} = \tilde{S}$, i.e., if $S_\mathcal{T}$ is the exact support set of the mixture, \eqref{eq:ERC_mixture_general} implies \eqref{eq:ERC_tropp}, and the result of the lemma holds. 
Therefore, we consider $\tilde{S}\subset S_{\mathcal{T}}$.

We begin with BP. Assume for the sake of contradiction that $S'$ is the support of an alternative solution to the minimization \eqref{eq:l_1_norm}. Since $S'\ne \tilde{S}$, it must contain at least one atom that does not belong to $\tilde{S}$. We analyze three possible scenarios separately: \\
\noindent 1) $S'\cap S_{\mathcal{T}}\neq \emptyset$ and $S' \cap S_{\mathcal{T}}^c \neq \emptyset$: For any $j \in S'\cap S_{\mathcal{T}}$, we have $\|\bK_{S_{\mathcal{T}}}^{\dagger}\bk_j\|_1 = 1$. On the other hand, if \eqref{eq:ERC_mixture_general} holds, then $\|\bK_{S_{\mathcal{T}}}^{\dagger}\bK_{S'\setminus S_{\mathcal{T}}}\|_{1,1}<1$. We also have that 
\begin{equation}
\label{eq:proof_2}
    \begin{aligned}
\|\bx_{S_{\mathcal{T}}}\|_1 =& \|\bK_{S_{\mathcal{T}}}^{\dagger}\bK_{\tilde{S}}\bx_{\tilde{S}}\|_1
 =\|\bx_{\tilde{S}}\|_1
 =\|\bK_{S_{\mathcal{T}}}^{\dagger}\bK_{S^{\prime}}\bx_{S^{\prime}}\|_1\\
 < & \|\bK_{S_{\mathcal{T}}}^{\dagger}\bK_{S^{\prime}}\|_{1,1}\|\bx_{S^{\prime}}\|_1\:,
    \end{aligned}
\end{equation}
where the second equality in \eqref{eq:proof_2} follows from the fact that $\tilde{S}\subset S_{\mathcal{T}}$, and the last strict inequality follows from \cite[Lemma 3.4]{tropp2004greed} since we have just shown that the $\ell_1$-norms of the columns of
$\bK_{S_{\mathcal{T}}}^{\dagger}\bK_{S^{\prime}}$ are not all identical. But, if \eqref{eq:ERC_mixture_general} holds as in the statement of the lemma, then we get from \eqref{eq:proof_2} that $\|\bx_{\tilde{S}}\|_1<\|\bx_{S^{\prime}}\|_1$, i.e., $S^\prime$ is not a minimizer, yielding a contradiction. \\
\noindent 2) $S'\cap S_{\mathcal{T}} = \emptyset$: 
In this case,
if \eqref{eq:ERC_mixture_general} holds, then $\|\bK_{S_{\mathcal{T}}}^{\dagger}\bK_{S^{\prime}}\|_{1,1}<1$. Following similar simplifications as in \eqref{eq:proof_2} (without the last strict inequality), we also get that $\|\bx_{\tilde{S}}\|_1<\|\bx_{S^{\prime}}\|_1$, yielding a contradiction. \\
3) $S'\subseteq S_{\mathcal{T}}$:  

In this case, we can write
\begin{equation}
\label{eq:proof_3}
    \begin{aligned}
    \by=\bK_{S_\mathcal{T}}\bar{\bx}_{{\tilde{S}}}  = \bK_{S_{\mathcal{T}}} \bar{\bx}_{S'} \:,
    \end{aligned}
\end{equation}
where $\bar{\bx}_{{\tilde{S}}}$ and $\bar{\bx}_{S'}$ represent the vectors $\bx_{{\tilde{S}}}$ and $\bx_{S'}$ restricted to $S_{\mathcal{T}}$. Therefore,
\begin{equation}
\label{eq:proof_4}
    \begin{aligned}
    \bK_{S_\mathcal{T}}\left(\bar{\bx}_{{\tilde{S}}} - \bar{\bx}_{S'}\right) = 0 \:,
    \end{aligned}
\end{equation}
implying that $\left(\bar{\bx}_{{\tilde{S}}} - \bar{\bx}_{S'}\right)$ must be in the null space of $\bK_{S_{\mathcal{T}}}$. However, per \eqref{eq:ERC_tropp}, $\bK_{S_{\mathcal{T}}}^H\bK_{S_{\mathcal{T}}}$ is invertible, thus  $\bK_{S_\mathcal{T}}$ is full rank, yielding a contradiction.

We conclude that \eqref{eq:l_1_norm} must yield the correct $\tilde{S}$ if condition \eqref{eq:ERC_mixture_general} in the statement of the lemma holds.
Next, we consider OMP. Let $\br_i$ denote the residual at the $i$-th iteration.
Since $\|\bA\mathbf{\beta}\|_{\infty}\leq \|\bA^H\|_{1,1}\|\mathbf{\beta}\|_{\infty}$, we have that 
\begin{equation}
\begin{aligned}
\label{eq:omp_proof_2}
       \|\bK_{S_{\mathcal{T}}}^{\dagger}\bK_{S_{\mathcal{T}}^c}\|_{1,1}
       \geq & \frac{\|\bK_{S_{\mathcal{T}}^c}^H(\bK_{S_{\mathcal{T}}}^{\dagger})^H \bK_{S_{\mathcal{T}}}^H \br_i\|_{\infty}}{\|\bK_{S_{\mathcal{T}}}^H \br_i\|_{\infty}}\\
      = &\frac{\|\bK_{S_{\mathcal{T}}^c}^H \br_i\|_{\infty}}{\|\bK_{S_{\mathcal{T}}}^H \br_i\|_{\infty}}\:.
       \end{aligned}
\end{equation}
The equality in \eqref{eq:omp_proof_2} is because $(\bK_{S_{\mathcal{T}}}^{\dagger})^H \bK_{S_{\mathcal{T}}}^H$ is the projection matrix onto the column space of $\bK_{S_{\mathcal{T}}}$. On the other hand, since $\tilde{S} \subset S_{\mathcal{T}}$, we can write $\by = \bK_{\tilde{S}}\bx_{\tilde{S}} = \bK_{S_{\mathcal{T}}}\bar{\bx}_{\tilde{S}}$, where $\bar{\bx}_{\tilde{S}}$ is the vector $\bx_{\tilde{S}}$ restricted to set $S_\mathcal{T}$, thus, $\by$ (and in turn $\br_i$) are in the column space of $\bK_{S_{\mathcal{T}}}$. From \eqref{eq:omp_proof_2}, it follows that if \eqref{eq:ERC_mixture_general} holds, then \begin{equation}
\label{eq:omp_proof_3}
 \frac{\|\bK_{S_{\mathcal{T}}^c}^H \br_i\|_{\infty}}{\|\bK_{S_{\mathcal{T}}}^H \br_i\|_{\infty}}< 1\:.
\end{equation}
\textcolor{black}{Hence, from \eqref{eq:omp_proof_3}, OMP will select one atom from the set $S_{\mathcal{T}}$ at iteration $i+1$. We consider two cases. First,  assume OMP selects $k = |\tilde{S}|$ atoms in the first $k$ iterations.
  Since the atoms in $\bK_{S_{\mathcal{T}}}$ are linearly independent, OMP recovers the optimal solution and stops. Second, let us assume OMP selects an incorrect atom from the set $S_{\mathcal{T}}$. Hence, there must exist an alternative solution $S'\subseteq S_{\mathcal{T}}$ such that $\bK_{S_{\mathcal{T}}}\bar{\bx}_{\tilde{S}} = \bK_{S_{\mathcal{T}}}\bar{\bx}_{S'}$, implying that  $\left(\bar{\bx}_{\tilde{S}}-\bar{\bx}_{S'}\right)$ is in the null space of $\bK_{S_{\mathcal{T}}}$. However, since $\bK_{S_{\mathcal{T}}}^H\bK_{S_{\mathcal{T}}}$ is invertible per \eqref{eq:ERC_mixture_general}, the matrix $\bK_{S_{\mathcal{T}}}$ is full rank, yielding a contradiction. 
  We conclude that if \eqref{eq:ERC_mixture_general} holds, then the program in \eqref{eq:l_1_norm} must recover the exact vector $\bx$ with support $\tilde{S}$.}
\end{proof}
\section{Proof of Theorem \ref{thr:guarantee_upperbound_mixture_strong}} 
\label{App:proof_strong_mixture}
The proof of Theorem \ref{thr:guarantee_upperbound_mixture_strong} and Theorem \ref{thr:guarantee_mixture_sparse_based} follow the same reasoning, so we only provide the proof for Theorem \ref{thr:guarantee_upperbound_mixture_strong}.
Let $\mathcal{T}\in \mathbb{Q}_k\left(m\right)$. Bounding the LHS of \eqref{eq:ERC_mixture_general},
 \begin{equation}
 \label{eq:proof_technique}
    \begin{aligned}
       \|\bK^{\dagger}_{S_{\mathcal{T}}}\bK_{S^c_{\mathcal{T}}}\|_{1,1} \leq
       \|\left(\bK_{S_{\mathcal{T}}}^H\bK_{S_{\mathcal{T}}}\right)^{-1}\|_{1,1}\|\bK_{S_{\mathcal{T}}}^H\bK_{S^c_{\mathcal{T}}}\|_{1,1}\:.
    \end{aligned}
\end{equation}
 We find an upper bound to the RHS of \eqref{eq:proof_technique}. To this end, we use the new notions of NPI and NPA introduced in Definition \ref{def:npi_coh}. For the second term on the RHS of \eqref{eq:proof_technique}, we can readily write
 \begin{equation}
 \|\bK_{S_{\mathcal{T}}}^H\bK_{S^c_{\mathcal{T}}}\|_{1,1}\leq \zeta_{k,m} \:.  
 \end{equation}
For the first term on the RHS of \eqref{eq:proof_technique}, and following the same reasoning as in \cite{tropp2004greed}, we write
\begin{equation}
\label{eq:proof_B_I}
       \bK_{S_{\mathcal{T}}}^H\bK_{S_{\mathcal{T}}} = \bB + \bI_{|S_{\mathcal{T}}|}\:,
\end{equation}
for some matrix $\bB$, where $\bB$ is all zero along the diagonal as the dictionary $\bK$ is normalized. From \eqref{eq:npa_coherence}, we conclude that $\|\bB\|_{1,1} \leq \nu_{k,m}$.
From the Neumman series, if $\|\bB\|_{1,1}<1$, $\sum_{k=0}^\infty \left(-\bB\right)^k$ converges to $\left(\bB + \bI_{|S_{\mathcal{T}}|}\right)^{-1}$ \cite{tropp2004greed},\cite{kreyszig1978introductory}. Hence,
%
\begin{equation}
    \begin{aligned}
        \|\left(\bK_{S_{\mathcal{T}}}^H\bK_{S_{\mathcal{T}}}\right)^{-1}\|_{1,1} &= \|\sum_{k=0}^\infty \left(-\bB\right)^k\|_{1,1}
       \leq \sum_{k=0}^\infty \|\bB\|^{k}_{1,1}\\
       &= \frac{1}{1-\|\bB\|_{1,1}} \leq \frac{1}{1-\nu_{k,m}}.
    \end{aligned}
\end{equation}
Therefore,
\begin{equation}
\label{eq:proof_final_npd_l1_upperbound_inequality-strong}
 \|\bK^{\dagger}_{S_{\mathcal{T}}}\bK_{S^c_{\mathcal{T}}}\|_{1,1} \leq \frac{\zeta_{k,m}}{1-\nu_{k,m}} \: .
\end{equation}
Hence, from \eqref{eq:ERC_mixture_general}, we conclude that, if \eqref{eq:guarantee_NPD_strong} holds then BP and OMP can exactly recover the sparse vector.

\section{Proof of Theorem \ref{thr:single_2_mixture}}
\label{App:proof_single_2_mixture}
We also use the inequality in \eqref{eq:proof_technique}. Recalling that $S_p$ is the support set corresponding to a single period $p$, 
for the second term on the RHS of \eqref{eq:proof_technique}, we have
\begin{equation}
     \begin{aligned}
\|\bK_{S_{\mathcal{T}}}^H\bK_{S^c_{\mathcal{T}}}\|_{1,1} = & \|\bK_{\left(\bigcup_{p\in \mathcal{T}}S_p\right)}^H\bK_{\left(\mathbb{N}\setminus \bigcup_{p\in \mathcal{T}}S_p\right)}\|_{1,1} \leq \sum_{p \in \mathcal{T}} \zeta_p,
    \end{aligned}
\end{equation}
where the inequality follows from the fact that the elements in $\mathcal{T}$ may share some divisors and that $\|\bK_{S_p}^H\bK_{S^c_p}\|_{1,1} \leq \zeta_p$ from the definition of $\zeta_p$ in  \eqref{eq:restricted_inter_coh_special}.
We rewrite the first term on the RHS of \eqref{eq:proof_technique},
\begin{equation}
       \bK_{\left(\bigcup_{p\in \mathcal{T}}S_p\right)}^H\bK_{\left(\bigcup_{p\in \mathcal{T}}S_p\right)} = \bB + \bI_{|S_{\mathcal{T}}|}.
\end{equation}
Without loss of generality, assume index $i^* \in S_{p_1}$ corresponds to $\|\bB\|_{1,1}$. In other words, $i^* = \argmax_{i\in S_{\mathcal{T}}} \sum_{j\neq i, j\in S_{\mathcal{T}}}  \left|\langle\bk_i,\bk_j\rangle \right|$. The contribution of atoms in $S_{p_1}$ to the previous sum of inner products is
\begin{equation}
\label{eq:proof_s2m_1}
 \sum_{\substack{j\in S_{p_1}\\j\neq {i^*}}}\left|\langle\bk_{i^*},\bk_j\rangle\right| \leq \nu_{p_1} \:.  
\end{equation}
The inner products between $\bk_i$ and the atoms in $S_{p_1}^c$ are in the summation in \eqref{eq:restricted_inter_coh_special} for $p_j \in \mathcal{T}\setminus p_{1}$. Hence, the contribution of these atoms satisfies
\begin{equation}
\label{eq:proof_s2m_2}
 \sum_{\substack{j \in S_{p_j}\\j\neq {i^*}} }\left|\langle\bk_{i^*},\bk_j\rangle\right| \leq \sum_{p_j\in \mathcal{T}\setminus p_1} \zeta_{p_j}\:.
\end{equation}
From \eqref{eq:proof_s2m_1} and \eqref{eq:proof_s2m_2} we get that $\|\bB\|_{1,1}  \leq \max_{j\in\mathcal{T}} \left(\nu_{p_j} + \sum_{p_i \in  \mathcal{T}\setminus p_j} \zeta_{p_i}\right)$. Following the same reasoning as in the proof of Theorem \ref{thr:guarantee_upperbound_mixture_strong}, 
if $\|\bB\|_{1,1}<1$, then 
\begin{equation}
    \begin{aligned}
       \|\left(\bK_{S_{\mathcal{T}}}^H \bK_{S_{\mathcal{T}}}\right)&^{-1}\|_{1,1}  \leq
        \frac{1}{1-\|\bB\|_{1,1}}\\ &\leq  \frac{1}{1-\max_{j\in\mathcal{T}} \left(\nu_{p_j} + \sum_{p_i \in  \mathcal{T}\setminus p_j} \zeta_{p_i}\right)}.
    \end{aligned}
\end{equation}
Hence, we conclude that
\begin{equation}
\label{eq:proof_single_2_mixture}
    \begin{aligned}
        \|\bK^{\dagger}_{S_{\mathcal{T}}}\bK_{S^c_{\mathcal{T}}}\|_{1,1}   \leq \frac{\sum_{p_j \in \mathcal{T}}\zeta_{p_j}}{1-\max_{j\in\mathcal{T}} \left(\nu_{p_j} + \sum_{p_i \in  \mathcal{T}\setminus p_j} \zeta_{p_i}\right)}.
    \end{aligned}
\end{equation}
From \eqref{eq:ERC_mixture_general}, we conclude that if the RHS of \eqref{eq:proof_single_2_mixture} is smaller than  $1$, then BP and OMP can recover the exact sparse representation of the periodic mixture whose hidden periods are in $\mathcal{T}$, or equivalently,
\begin{equation}
    \begin{aligned}
  &\sum_{p_j \in \mathcal{T}}\zeta_{p_j} +\max_{p_j\in\mathcal{T}} \left(\nu_{p_j} + \sum_{p_i \in  \mathcal{T}\setminus p_j} \zeta_{p_i}\right) \leq \\
        &\hspace{2cm} 2\sum_{p_j \in \mathcal{T}}\zeta_{p_j} + \max_{p_j\in\mathcal{T}}\nu_{p_j} - \min_{p_j\in\mathcal{T}} \zeta_{p_j} < 1.
    \end{aligned}
\end{equation}

\section{Proof of Theorem \ref{thr:noise_bounded_general_zeta_nu}} 
\label{App:proof_noise_bounded_general_zeta_nu}
First, we assert the following lemma.
\begin{lemma}
\label{lemma:lambda_min}
Let $\mathcal{T} \in \mathbb{Q}_k\left(m\right)$, $\lambda_{\min}$ is the minimum eigenvalue of $\bK_{S_{\mathcal{T}}}^T\bK_{S_{\mathcal{T}}}$ and $\nu_{k,m}$ as in \eqref{eq:npa_coherence}. 
If $\nu_{k,m}<1$, then $\lambda_{\min} \geq 1-\nu_{k,m}$.
\end{lemma}
\begin{proof}
To prove this result, we follow the reasoning in \cite[Lemma 2]{cai2011orthogonal}. To find a lower bound on the minimum eigenvalue, we can show that, when $\nu_{k,m} < 1$, the matrix $\bK_{S_{\mathcal{T}}}^T\bK_{S_{\mathcal{T}}}-\lambda\bI$ is nonsingular for any $\lambda<1- \nu_{k,m}$. Equivalently, we show that $\left(\bK_{S_{\mathcal{T}}}^T\bK_{S_{\mathcal{T}}}-\lambda\bI\right)\bc \neq \mathbf{0}$, where $\bc=\left[c_1,c_2,\ldots,c_s\right]^T\ne \mathbf{0}$, 
and without loss of generality, we assume $|c_1|\geq|c_2| \geq \ldots |c_s|$. Let $\bK_{S_{\mathcal{T}}} = \left[\bk_{S_{\mathcal{T}_1}},\bk_{S_{\mathcal{T}_2}},\ldots,\bk_{S_{\mathcal{T}_s}}\right]$ and $s\leq k$. Then, for the first coordinate of $\left(\bK_{S_{\mathcal{T}}}^T\bK_{S_{\mathcal{T}}}-\lambda \bI\right)\bc$, we have that
\begin{equation}
\label{eq:appendix_lemma_proof}
    \begin{aligned}
         &\left|\Big[\left(\bK_{S_{\mathcal{T}}}^T\bK_{S_{\mathcal{T}}}-\lambda\bI\right)\bc\Big]_1\right| \\
        =&\left|\left(1-\lambda\right)c_1 + \bk_{S_{\mathcal{T}_1}}^T\bk_{S_{\mathcal{T}_2}}c_2 +\ldots+ \bk_{S_{\mathcal{T}_1}}^T\bk_{S_{\mathcal{T}_s}}c_s\right|\\
        \geq & \left|\left(1-\lambda\right)c_1\right| -\left| \bk_{S_{\mathcal{T}_1}}^T\bk_{S_{\mathcal{T}_2}}c_2\right| - \ldots -  \left|\bk_{S_{\mathcal{T}_1}}^T\bk_{S_{\mathcal{T}_s}}c_s\right|\\
         \geq & \left(1-\lambda\right)\left|c_1\right| -\left(\left| \bk_{S_{\mathcal{T}_1}}^T\bk_{S_{\mathcal{T}_2}}\right|+ \ldots +  \left|\bk_{S_{\mathcal{T}_1}}^T\bk_{S_{\mathcal{T}_s}}\right|\right) \left|c_2\right|\\
        \geq &  \left(1-\lambda\right)\left|c_1\right| -\nu_{k,m} \left|c_2\right|.
    \end{aligned}
\end{equation}
Therefore, for any $\lambda<1 - \nu_{k,m}$, 
\begin{equation}
        \left|\Big[\left(\bK_{S_{\mathcal{T}}}^T\bK_{S_{\mathcal{T}}}-\lambda\bI\right)\bc\Big]_1\right| > \nu_{k,m}\left|c_1\right| - \nu_{k,m}\left|c_2\right|\geq 0\\.
\end{equation}
Hence, we conclude that $\lambda_{\min} \geq 1-\nu_{k,m}$.
\end{proof}
We proceed with the proof of the theorem. Following \cite[Proposition 1]{cai2011orthogonal}, for a given set $\tilde{S} \in S_{\mathcal{T}}$, if \eqref{eq:ERC_tropp} holds,
then OMP can recover the exact support set $\tilde{S}$, if all the nonzero entries of $\bx$ obey \eqref{eq:cai_noise_bounded}.

\begin{equation}
    \label{eq:cai_noise_bounded}
    |x_i|\geq \frac{2\epsilon}{\left(1-M\right)\lambda_{\min}}.
\end{equation}
In \eqref{eq:cai_noise_bounded}, $\lambda_{\min}$ is the minimum eigenvalue of the Gram matrix $\bK_{\tilde{S}}^T\bK_{\tilde{S}}$. From Definition \ref{def:npi_coh}, $M \leq \zeta_{k,m} + \nu_{k,m}$. Hence, invoking Lemma \ref{lemma:lambda_min}, we can express the condition in \eqref{eq:cai_noise_bounded} as
\begin{equation}
    |x_i|\geq \frac{2\epsilon}{\left(1-\zeta_{k,m}-\nu_{k,m}\right)\left(1-\nu_{k,m}\right)}.
\end{equation}
The proof for Theorem \ref{thr:noise_bounded_general_s2m} follows the same reasoning as Theorem \ref{thr:noise_bounded_general_zeta_nu}, thus is omitted for brevity. 
\section{Proof of Theorem 
\ref{thr:noise_gaussian_s2m}}
\label{App:proof_noise_gaussian}

To prove this result we need the following lemma from \cite{cai2011orthogonal}.
\begin{lemma}[Lemma 3, From \cite{cai2011orthogonal} and \cite{cai2009recovery}]
\label{lemma:noise}

Let 
\begin{equation}
\label{eq:noise_set_B2}
B_2 = \{ \bw : \|\bw\|_2 \leq \sigma \sqrt{L + 2 \sqrt{L \log L}}\}\:.
\end{equation}The Gaussian error $\bw \sim \mathcal{N}\left(\mathbf{0},\sigma^2\bI_n\right)$ satisfies
\begin{equation}
    P\left(\bw \in B_2\right) \geq 1 - \frac{1}{L} \:.
\end{equation}
\end{lemma}
Given the definition of $B_2$, the result follows as a direct consequence of Theorem \ref{thr:noise_bounded_general_s2m} and Lemma \ref{lemma:noise}.
\end{document}